\documentclass[10pt,twocolumn,journal]{IEEEtran}
\AtBeginDocument{\fontsize{9.8pt}{11.5pt}\selectfont}

\usepackage{amsmath,amssymb,amsfonts,amsthm,mathtools,bm}
\usepackage{braket}
\usepackage{algorithm}
\usepackage{algorithmic}
\usepackage{array}
\usepackage{multirow}
\usepackage{booktabs,tabularx,makecell}
\usepackage[caption=false,font=footnotesize]{subfig}
\usepackage[inline,shortlabels]{enumitem}
\usepackage{textcomp}
\usepackage{stfloats}
\usepackage{url}
\usepackage{verbatim}
\usepackage{graphicx}
\usepackage{cite}
\usepackage{microtype}
\usepackage{xcolor}
\usepackage[colorlinks=true,allcolors=blue]{hyperref}
\usepackage{orcidlink}

\theoremstyle{remark}

\usepackage{tikz}
\usepackage{pgfplots}
\pgfplotsset{compat=1.17}
\usetikzlibrary{arrows,shapes,positioning,calc,patterns,decorations.pathreplacing}
\usetikzlibrary{fit,backgrounds}
\usetikzlibrary{plotmarks}

\newtheorem{theorem}{Theorem}

\theoremstyle{definition}

\theoremstyle{remark}

\definecolor{qstDenseColor}{HTML}{4C4C4C}
\definecolor{qstLRColor}{HTML}{377EB8}
\definecolor{qstRGDColor}{HTML}{984EA3}
\definecolor{qstFWColor}{HTML}{FF7F00}
\definecolor{qstPropColor}{HTML}{1B9E77}

\DeclareRobustCommand{\qstDenseMark}{%
\tikz[baseline=-0.5ex]{
\draw[qstDenseColor,line width=0.85pt] (0,0)--(0.70,0);
\draw[qstDenseColor,only marks,mark=*,mark options={fill=qstDenseColor,draw=black,line width=0.25pt,scale=1.15}]
plot coordinates {(0.35,0)};
}}

\DeclareRobustCommand{\qstLRMark}{%
\tikz[baseline=-0.5ex]{
\draw[qstLRColor,line width=0.85pt] (0,0)--(0.70,0);
\draw[qstLRColor,only marks,mark=square*,mark options={fill=qstLRColor,draw=black,line width=0.25pt,scale=1.15}]
plot coordinates {(0.35,0)};
}}

\DeclareRobustCommand{\qstRGDMark}{%
\tikz[baseline=-0.5ex]{
\draw[qstRGDColor,line width=0.85pt] (0,0)--(0.70,0);
\draw[qstRGDColor,only marks,mark=triangle*,mark options={fill=qstRGDColor,draw=black,line width=0.25pt,scale=1.25}]
plot coordinates {(0.35,0)};
}}

\DeclareRobustCommand{\qstFWMark}{%
\tikz[baseline=-0.5ex]{
\draw[qstFWColor,line width=0.85pt] (0,0)--(0.70,0);
\draw[qstFWColor,only marks,mark=diamond*,mark options={fill=qstFWColor,draw=black,line width=0.25pt,scale=1.25}]
plot coordinates {(0.35,0)};
}}

\DeclareRobustCommand{\qstPropMark}{%
\tikz[baseline=-0.5ex]{
\draw[qstPropColor,line width=0.85pt] (0,0)--(0.70,0);
\draw[qstPropColor,only marks,mark=+,mark options={qstPropColor,line width=1.0pt,scale=1.55}]
plot coordinates {(0.35,0)};
}}

\DeclareRobustCommand{\qstTermMark}{%
\tikz[baseline=-0.5ex]{
\draw[black,line width=1.0pt] (0,0)--(0.22,0.22);
\draw[black,line width=1.0pt] (0,0.22)--(0.22,0);
}}

\DeclareRobustCommand{\qstCapLine}{%
\tikz[baseline=-0.5ex]{
\draw[black,dashed,line width=0.9pt] (0,0)--(0.70,0);
}}

\definecolor{algAtomColor}{HTML}{7F7F7F}
\definecolor{algCoeffColor}{HTML}{2CA02C}
\definecolor{algOracleColor}{HTML}{9467BD}
\definecolor{algFullColor}{HTML}{1F77B4}
\definecolor{algErrorColor}{HTML}{8B0000}
\definecolor{algKstarColor}{HTML}{003366}

\DeclareRobustCommand{\algAtomMark}{%
\tikz[baseline=-0.5ex]{
\draw[algAtomColor,dotted,line width=0.9pt] (0,0)--(0.70,0);
\draw[algAtomColor,only marks,mark=*,mark options={fill=algAtomColor,draw=black,line width=0.25pt,scale=1.15}]
plot coordinates {(0.35,0)};
}}

\DeclareRobustCommand{\algCoeffMark}{%
\tikz[baseline=-0.5ex]{
\draw[algCoeffColor,dashed,line width=0.9pt] (0,0)--(0.70,0);
\draw[algCoeffColor,only marks,mark=square*,mark options={fill=algCoeffColor,draw=black,line width=0.25pt,scale=1.15}]
plot coordinates {(0.35,0)};
}}

\DeclareRobustCommand{\algOracleMark}{%
\tikz[baseline=-0.5ex]{
\draw[algOracleColor,dash pattern=on 4pt off 1.4pt on 1pt off 1.4pt,line width=0.9pt] (0,0)--(0.70,0);
\draw[algOracleColor,only marks,mark=diamond*,mark options={fill=algOracleColor,draw=black,line width=0.25pt,scale=1.25}]
plot coordinates {(0.35,0)};
}}

\DeclareRobustCommand{\algFullMark}{%
\tikz[baseline=-0.5ex]{
\draw[algFullColor,line width=0.9pt] (0,0)--(0.70,0);
\draw[algFullColor,only marks,mark=triangle*,mark options={fill=algFullColor,draw=black,line width=0.25pt,scale=1.25}]
plot coordinates {(0.35,0)};
}}

\DeclareRobustCommand{\algRefactorLine}{%
\tikz[baseline=-0.5ex]{
\draw[black,dashed,line width=0.9pt] (0,0)--(0.70,0);
}}

\DeclareRobustCommand{\algErrorLine}{%
\tikz[baseline=-0.5ex]{
\draw[algErrorColor,dash pattern=on 4pt off 1.4pt on 1pt off 1.4pt,line width=0.9pt] (0,0)--(0.70,0);
}}

\DeclareRobustCommand{\algKstarLine}{%
\tikz[baseline=-0.5ex]{
\draw[algKstarColor,dotted,line width=1.05pt] (0,0)--(0.70,0);
}}
\title{Rank-Adaptive Matrix-Free Atomic \\ Quantum State Tomography}

\author{%
Amirhossein Taherpour~\orcidlink{0000-0003-4647-102X},%
\ Alireza Sadeghi~\orcidlink{0000-0003-1280-7592},%
\ and Georgios B.~Giannakis~\orcidlink{0000-0002-0196-0260},~\IEEEmembership{Fellow,~IEEE}%
\thanks{Amirhossein Taherpour, Alireza Sadeghi, and Georgios B.~Giannakis are with the Dept. of ECE, Univ. of Minnesota, Minneapolis, MN 55455, USA. E-mails: \{taher054, sadeg012, georgios\}@umn.edu. Part of this research was supported by NSF grants 2212318 and 2312547.} % \textit{Corresponding author: G. B.~Giannakis.}}%
}

\IEEEaftertitletext{\vspace{-0.8cm}}
\begin{document}

\maketitle
\begin{abstract}
Quantum state tomography estimates an unknown density operator from
measurement data. Dense reconstruction however, can be impractical for many-qubit
systems because the Hilbert-space dimension grows exponentially. This contribution
develops a rank-adaptive matrix-free approach to low-rank quantum state
tomography based on rank-one atomic coordinates. The density operator is
represented as a convex combination of pure-state atoms, which preserves
positivity and unit trace while avoiding dense density, measurement, and
gradient matrices. The resultant algorithm combines atom updates, simplex-constrained
coefficient reweighting, and periodic spectral refactorization using only
predicted probability vectors, atom vectors, and descriptor-level measurement
actions. These computations decompose across measurement outcomes, and admit a
master--worker implementation with the measurements partitioned across
workers. A rank penalty adapts the representation size during optimization.
Provable feasibility, prediction consistency, monotone descent of the
penalized objective, and an exact spectral proximal characterization of the
refactorization step are established. Simulations with Pauli measurements show favorable
accuracy--runtime--memory tradeoffs relative to competing alternatives.
\end{abstract}


\vspace{-0.2 cm}
\section{Introduction}
\label{sec:introduction}

Quantum state tomography (QST) estimates an unknown density operator from
measurement data \cite{NielsenChuang2010QuantumComputation,Hradil1997QuantumStateEstimation}.
It is central to quantum-device characterization, but scalability is a key challenge.
For \(N\) qubits, the Hilbert-space dimension is \(D=2^N\), and a generic
density matrix has \(O(D^2)\) entries. Scalable QST must therefore leverage structure in the state, the measurements, or the reconstruction algorithm~\cite{taherpour2026distributed}.

One broad direction exploits low rank. Compressed-sensing and
sample-complexity studies show that nearly pure or low-rank states can be
recovered from fewer measurements
\cite{Gross2010CSQST,Flammia2012CSQST,Haah2017SampleOptimalQST}. On the
algorithmic side, dense maximum-likelihood and projected-gradient methods have been investigated
\cite{Bolduc2017PGDQST}. Low-rank projected, factored, and
Riemannian methods reduce the number of variables by constraining or
parametrizing the rank
\cite{Burer2003LowRankSDP,Kyrillidis2018NonconvexQST,Hsu2024RGDQST,Wang2024EfficientFGD},
and stochastic or online variants further improve update efficiency
\cite{Gaikwad2025GDQST,Cai2025OnlineSGDQST}. These methods improve
scalability, but they typically still assume a fixed rank or rely on dense
matrix and spectral operations.

A second direction imposes stronger models on the state or changes the
reconstruction target. Tensor-network tomography uses structured
representations such as matrix-product or locally purified models
\cite{Cramer2010EfficientQST,Baumgratz2013ScalableDensity,Cai2026MPOQST}.
Neural-network methods use expressive variational models
\cite{Torlai2018NNQST}, and classical-shadow methods estimate many observables
without necessarily producing an explicit low-rank density estimate
\cite{Huang2020ClassicalShadows}. Rank selection has also been studied through
penalized and model-selection criteria
\cite{Alquier2013RankPenalizedQST,Guta2012RankModelQST}. These directions are
powerful, but they do not simultaneously yield explicit low-rank density reconstruction, matrix-free operations, and adaptive rank control.

This paper bridges the gap by developing a rank-adaptive matrix-free method
for low-rank QST based on rank-one atomic coordinates. The density operator is
represented as a convex combination of pure-state atoms, which preserves
positivity and unit trace by construction. For fixed atoms, the predicted
probabilities are linear in the mixture weights. This enables objective
evaluation and updates using atom responses and descriptor-level measurement
actions, without forming dense density, measurement, or gradient matrices. The
algorithm combines atom updates, coefficient reweighting, and periodic
spectral refactorization based on Ritz candidates
\cite{Lanczos1950Iteration,Saad2011NumericalMethods}. The penalty allows rank to grow or shrink during optimization and links the method to sparse semidefinite and Frank--Wolfe iterations~\cite{Hazan2008SparseSDP},~\cite{Jaggi2013FrankWolfe}.

The main contributions are as follows. C1) We formulate low-rank QST in
atomic coordinates that preserve the physical density constraints while
avoiding dense \(D\times D\) storage. C2) We develop a measurement-partitioned
matrix-free algorithm that updates atoms, weights, and the rank
using predicted probability vectors, atom vectors, and descriptor-level measurement
actions only. C3) We prove feasibility, prediction consistency, monotone
descent of the penalized objective, and an exact spectral proximal
characterization of the refactorization step. Simulations with Pauli
measurements demonstrate favorable accuracy--runtime--memory tradeoffs relative to existing approaches.

\noindent{\it Notation.}
Lowercase letters denote scalars, bold lowercase letters vectors, and bold
uppercase letters matrices; \(\boldsymbol\rho\) denotes a density operator.
Hats denote empirical quantities, e.g., \(\widehat{\boldsymbol\pi}\). For a
closed set \(\mathcal C\), \(\operatorname{proj}_{\mathcal C}(\cdot)\)
denotes Euclidean projection.
For an index set \(\mathcal I\), \([\boldsymbol z]_{\mathcal I}\) denotes the
subvector of \(\boldsymbol z\) indexed by \(\mathcal I\).

The rest of the paper is organized as follows. Section~\ref{sec:formulation}
introduces the low-rank atomic QST. Section~\ref{sec:distributed-matrix-free-qst}
develops the distributed matrix-free algorithm and its analysis.
Section~\ref{sec:simulations} presents numerical tests. Finally,
Section~\ref{sec:conclusion} concludes the paper.

% \vspace{-0.25 cm}
\section{Low-Rank Atomic Tomography}
\label{sec:formulation}

\begin{figure*}[t]
    \centering
    \includegraphics[width=\textwidth]{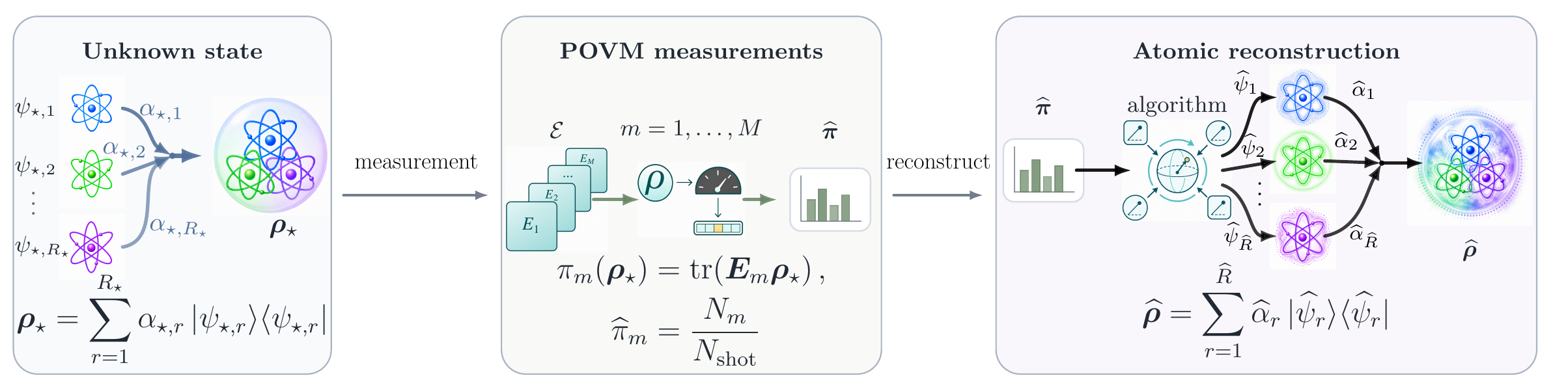}
    \vspace{-0.5 cm}
    \caption{Overview of atomic state tomography. An unknown quantum state is represented as a convex combination of pure-state atoms, measured through a POVM to yield empirical outcome frequencies, and it is reconstructed by estimating the active atoms and their mixture weights.}
    \vspace{-0.5 cm}
    \label{fig:atomic}
\end{figure*}

Consider an \(N\)-qubit quantum system with Hilbert space
\(\mathcal H\cong(\mathbb{C}^{2})^{\otimes N}\), having dimension \(D=2^N\).
A pure state can be represented by a unit-norm vector
\(\ket{\boldsymbol\psi}\in\mathcal H\), or equivalently by the rank-one
density operator
\(\ket{\boldsymbol\psi}\bra{\boldsymbol\psi}\). More generally, pure and
mixed states are described by density operators
\(\boldsymbol\rho\in\mathbb C^{D\times D}\) belonging to the feasible set
\cite{NielsenChuang2010QuantumComputation}
\begin{equation}
\label{eq:density-set}
\mathcal D
\coloneqq
\left\{
\boldsymbol\rho\in\mathbb C^{D\times D}
:
\boldsymbol\rho\succeq 0,\ 
\operatorname{tr}(\boldsymbol\rho)=1
\right\}.
\end{equation}
Pure states are the rank-one elements of \(\mathcal D\), whereas mixed states
are general elements of \(\mathcal D\). In this work, we focus on low-rank
density operators, which arise when the state can be represented as a mixture
of only a few pure-state atoms. For an active atom count \(R\ll D\), let
\[
\boldsymbol\Psi
:=
[\boldsymbol\psi_1,\ldots,\boldsymbol\psi_R]
\in\mathbb C^{D\times R},
\qquad
\|\boldsymbol\psi_r\|_2=1,\quad r=1,\ldots,R
\]
and let
\(\boldsymbol\alpha\in\boldsymbol\Delta_R\), where
\(\boldsymbol\Delta_R\) denotes the \(R\)-dimensional probability simplex
\begin{equation}
\nonumber
\boldsymbol\Delta_R
\coloneqq
\left\{
\boldsymbol\alpha\in\mathbb R^R_+
:
\alpha_r\ge 0,\ 
\sum_{r=1}^{R}\alpha_r=1
\right\}.
\end{equation}
The density operator then admits the atomic representation
\begin{equation}
\label{eq:atomic-density}
\boldsymbol\rho(\boldsymbol\alpha,\boldsymbol\Psi)
=
\sum_{r=1}^{R}
\alpha_r
\ket{\boldsymbol\psi_r}\bra{\boldsymbol\psi_r}
=
\boldsymbol\Psi
\operatorname{diag}(\boldsymbol\alpha)
\boldsymbol\Psi^* .
\end{equation}
This construction ensures
\begin{equation}
\label{eq:atomic-feasibility}
\nonumber
\boldsymbol\rho(\boldsymbol\alpha,\boldsymbol\Psi)\succeq0,
\quad
\operatorname{tr}
\left(
\boldsymbol\rho(\boldsymbol\alpha,\boldsymbol\Psi)
\right)=1,
\quad
\operatorname{rank}
\left(
\boldsymbol\rho(\boldsymbol\alpha,\boldsymbol\Psi)
\right)\le R .
\end{equation}
Thus, \(R\) is the number of active atoms and an upper bound on the matrix rank
of the represented density operator.

Our objective is to recover an unknown low-rank density operator
\(\boldsymbol\rho\in\mathcal D\) from measurement data; see
Fig.~\ref{fig:atomic}. Measurements are described by a positive
operator-valued measure (POVM)
\(\mathcal E :=\{\boldsymbol E_m\}_{m=1}^{M}\), where each
\(\boldsymbol E_m\) corresponds to one possible outcome
\(y\in\{1,\ldots,M\}\), satisfies
\(\boldsymbol E_m\succeq0\), and the POVM elements obey
\(\sum_{m=1}^{M}\boldsymbol E_m=\boldsymbol I_D\). The Born rule implies
\cite{NielsenChuang2010QuantumComputation}
\begin{equation}
\mathbb P
\left(
y=m\mid \boldsymbol\rho,\mathcal E
\right)
=
\operatorname{tr}
\left(
\boldsymbol E_m\boldsymbol\rho
\right),
\qquad
m=1,\ldots,M .
\end{equation}
For a fixed POVM \(\mathcal E\), define the Born probability map
\begin{equation}
\label{eq:born-map}
\boldsymbol\pi(\boldsymbol\rho)
:=
\left[
\pi_1(\boldsymbol\rho),
\ldots,
\pi_M(\boldsymbol\rho)
\right]^\top,
\qquad
\pi_m(\boldsymbol\rho)
\coloneqq
\operatorname{tr}
\left(
\boldsymbol E_m\boldsymbol\rho
\right).
\end{equation}
Since \(\boldsymbol\rho\succeq0\), \(\operatorname{tr}(\boldsymbol\rho)=1\),
\(\boldsymbol E_m\succeq0\), and
\(\sum_{m=1}^{M}\boldsymbol E_m=\boldsymbol I_D\), the vector
\(\boldsymbol\pi(\boldsymbol\rho)\) is a valid probability mass function, that is
\(\boldsymbol\pi(\boldsymbol\rho)\in\boldsymbol\Delta_M\). Thus, the measurement
outcome is modeled as
\begin{equation}
y\sim\mathrm{Categorical}
\left(
\boldsymbol\pi(\boldsymbol\rho)
\right).
\end{equation}
Given \(N_{\rm shot}\) independent outcomes
\(\{y_i\}_{i=1}^{N_{\rm shot}}\), let
\[
N_m=\sum_{i=1}^{N_{\rm shot}}\mathbf{1}\{y_i=m\},
\qquad m=1,\ldots,M .
\]
The empirical probability vector is
\begin{equation}
\label{eq:empirical-pmf}
\widehat{\pi}_m
=
\frac{N_m}{N_{\rm shot}},
\qquad
m=1,\ldots,M
\end{equation}
which is the maximum-likelihood estimator of the categorical probabilities
\cite{Hradil1997QuantumStateEstimation}. If several measurement settings are
used, the index \(m\) can be understood as a combined setting--outcome index,
so that the full experiment is still represented by a single POVM.

Given \(\widehat{\boldsymbol\pi}\) and \(\mathcal E\), the goal is to estimate
a low-rank density operator whose Born probabilities explain the empirical
frequencies. Equivalently, under \eqref{eq:atomic-density}, the goal is to
estimate an active atom count \(\widehat R\), mixture weights
\(\widehat{\boldsymbol\alpha}\in\boldsymbol\Delta_{\widehat R}\), and atoms
\[
\widehat{\boldsymbol\Psi}
=
[
\widehat{\boldsymbol\psi}_1,\ldots,
\widehat{\boldsymbol\psi}_{\widehat R}
],
\]
to obtain the density estimate
\begin{equation}
\label{eq:qst-estimator}
\widehat{\boldsymbol\rho}
=
\sum_{r=1}^{\widehat R}
\widehat\alpha_r
\ket{\widehat{\boldsymbol\psi}_r}
\bra{\widehat{\boldsymbol\psi}_r}.
\end{equation}

We fit the Born probabilities using the \(\varepsilon\)-stabilized empirical
negative log-likelihood, or cross-entropy surrogate
\begin{equation}
\label{eq:qst-loss}
\mathcal L_\varepsilon(\boldsymbol\pi)
\coloneqq
-
\sum_{m=1}^{M}
\widehat{\pi}_m
\log(\pi_m+\varepsilon),
\qquad
\boldsymbol \pi\in\boldsymbol\Delta_M
\end{equation}
where \(\varepsilon>0\) avoids singular logarithms when predicted
probabilities vanish. When \(\varepsilon=0\), this reduces to the usual
empirical negative log-likelihood for categorical observations.

For each atom \(\ket{\boldsymbol\psi_r}\), define its response to the POVM by
\begin{equation}
\label{eq:atom-response}
a_m(\boldsymbol\psi_r)
\coloneqq
\bra{\boldsymbol\psi_r}
\boldsymbol E_m
\ket{\boldsymbol\psi_r},
\qquad
m=1,\ldots,M .
\end{equation}
Collecting these responses gives
\[
\boldsymbol a(\boldsymbol\psi_r)
:=
[
a_1(\boldsymbol\psi_r),
\ldots,
a_M(\boldsymbol\psi_r)
]^\top .
\]
Because \(\mathcal E\) is a POVM, each atom response is itself a probability
mass vector
\[
\boldsymbol a(\boldsymbol\psi_r)\in\boldsymbol\Delta_M,
\qquad r=1,\ldots,R .
\]
For the atom matrix \(\boldsymbol\Psi\), define the atom-response matrix
\begin{equation}
\label{eq:atom-measurement-matrix}
\boldsymbol A(\boldsymbol\Psi)
:=
[
\boldsymbol a(\boldsymbol\psi_1),
\ldots,
\boldsymbol a(\boldsymbol\psi_R)
],
\qquad
\boldsymbol A(\boldsymbol\Psi)\in\mathbb R_+^{M\times R}.
\end{equation}
Substituting the density model \eqref{eq:atomic-density} into the Born
map \eqref{eq:born-map} yields
\begin{equation}
\label{eq:atomic-born-map}
\boldsymbol\pi
\left(
\boldsymbol\rho(\boldsymbol\alpha,\boldsymbol\Psi)
\right)
=
\boldsymbol A(\boldsymbol\Psi)\boldsymbol\alpha .
\end{equation}
Thus, for fixed atoms, the predicted measurement probabilities are a
simplex-weighted mixture of atom responses.

To allow a data-adaptive rank selection, we
choose a rank budget \(R_{\max}\), and use the penalized objective
\begin{equation}
\label{eq:fixed-rank-objective}
J_R(\boldsymbol\alpha,\boldsymbol\Psi)
\coloneqq
\mathcal L_\varepsilon
\left(
\boldsymbol A(\boldsymbol\Psi)\boldsymbol\alpha
\right)
+
\mu R
\end{equation}
where \(\mu>0\) controls the cost of adding rank-one atoms. The penalty
therefore acts on the active atom count, which upper-bounds the rank of the
resulting density operator. Similar rank-penalized and rank-selection ideas
have been used in low-rank quantum-state estimation
\cite{Alquier2013RankPenalizedQST,Guta2012RankModelQST}. The resulting
rank-adaptive atomic tomography problem is
\begin{align}
\label{eq:rank-adaptive-atomic-qst}
(\widehat R,\widehat{\boldsymbol\alpha},\widehat{\boldsymbol\Psi})
\in
\operatorname*{arg\,min}_{R,\boldsymbol\alpha,\{\ket{\boldsymbol\psi_r}\}_{r=1}^{R}}
&\quad
J_R(\boldsymbol\alpha,\boldsymbol\Psi)
\\
\text{s.t.}\qquad
&\quad
1\le R\le R_{\max},
\nonumber\\
&\quad
\boldsymbol\alpha\in\boldsymbol\Delta_R,
\nonumber\\
&\quad
\|\boldsymbol\psi_r\|_2=1,
\qquad r=1,\ldots,R .
\nonumber
\end{align}

Solving \eqref{eq:rank-adaptive-atomic-qst} is challenging because it couples the
discrete \(R\) with the continuous variables
\(\boldsymbol\alpha\) and \(\boldsymbol\Psi\), and remains nonconvex in the 
atoms even for fixed \(R\), since each atom enters through
\(\ket{\boldsymbol\psi_r}\bra{\boldsymbol\psi_r}\). While low-rank and
compressed-sensing QST rely on the fact that nearly pure states can be learned
from fewer measurements than generic dense states
\cite{Gross2010CSQST,Flammia2012CSQST}, our setting must also maintain a
matrix-free atomic representation and adapt the active rank.

The next section develops a distributed matrix-free algorithm for
\eqref{eq:rank-adaptive-atomic-qst}. It uses
\eqref{eq:atomic-born-map}, \eqref{eq:fixed-rank-objective}, and
\eqref{eq:likelihood-ratio-vector}--\eqref{eq:negative-gradient-product} to
update atoms, reweight coefficients, and perform rank-adaptive refactorization
without forming dense matrices.

%%%%%%%%%%%%%%%%%%%%%%%%%%%%%%%%%%%%%%%%%%%%%%%%%%%%%%%%%%%%%%%%%%%%%%%%%%%%%%%%%%%%
\vspace{-0.25 cm}
\section{Distributed Matrix-Free Low-Rank \\ Quantum State Tomography}
\label{sec:distributed-matrix-free-qst}

% We solve \eqref{eq:rank-adaptive-atomic-qst} using distributed master-workers framework. 
The main idea here is to solve \eqref{eq:rank-adaptive-atomic-qst} using distributed master-workers framework. The optimization variables including rank $R_t$, mixture weights $\boldsymbol{\alpha}_t$, and the collection of atoms $\boldsymbol{\Psi}_t$, gathered in \(\mathcal X_t=(R_t,\boldsymbol\alpha_t,\boldsymbol\Psi_t)\) at the master node, while the workers handle the expensive measurement-dependent computations. 
While workers provide update to master node, the master decides if an update should be accepted or not. 

To solve \eqref{eq:rank-adaptive-atomic-qst}, we use alternating updates for
the continuous variables together with a slower rank-adaptation mechanism. Let us define the likelihood-ratio vector
\begin{equation}
\label{eq:likelihood-ratio-vector}
\boldsymbol c(\boldsymbol\alpha,\boldsymbol\Psi)
\coloneqq
\widehat{\boldsymbol\pi}
\oslash
\left(
\boldsymbol A(\boldsymbol\Psi)\boldsymbol\alpha
+
\varepsilon\boldsymbol 1
\right),
\end{equation}
where \(\oslash\) denotes componentwise division. For fixed
\(\boldsymbol\Psi\), the objective in \eqref{eq:rank-adaptive-atomic-qst} is convex in \(\boldsymbol\alpha\), and
the chain rule gives the gradient with respect to the mixture weights as
\begin{equation}
\label{eq:weight-gradient}
\nabla_{\boldsymbol\alpha}
J_R(\boldsymbol\alpha,\boldsymbol\Psi)
=
-
\boldsymbol A(\boldsymbol\Psi)^\top
\boldsymbol c(\boldsymbol\alpha,\boldsymbol\Psi).
\end{equation}
Using this gradient together with projection onto the simplex yields a
projected-gradient update for \(\boldsymbol\alpha\). 
For fixed \(R\) and \(\boldsymbol\alpha\), one can also update the active atoms
using the gradient of the loss with respect to each atom. In Wirtinger form,
the gradient with respect to the conjugate atom variable is
\begin{equation}
\label{eq:atom-wirtinger-gradient}
\nabla_{\boldsymbol\psi_r^*}
J_R(\boldsymbol\alpha,\boldsymbol\Psi)
=
-
\alpha_r
\boldsymbol G(\boldsymbol\alpha,\boldsymbol\Psi)
\boldsymbol\psi_r,
\qquad r=1,\ldots,R .
\end{equation}
Under the unit-norm constraint
\(\|\boldsymbol\psi_r\|_2=1\), the projected descent direction on the tangent
space of the complex sphere~is
\begin{equation}
\label{eq:atom-projected-descent-direction}
\boldsymbol d_r
=
\alpha_r
\left(
\boldsymbol I_D
-
\boldsymbol\psi_r\boldsymbol\psi_r^*
\right)
\boldsymbol G(\boldsymbol\alpha,\boldsymbol\Psi)
\boldsymbol\psi_r .
\end{equation}
This descent direction is used later for the atom updates.

The more challenging part remains the update of the discrete rank \(R\). It is a
discrete model-order variable, not admitting gradient-based updates. Instead,
we update \(R\) on a slower outer time-scale using a rank-selection mechanism.
The details of this procedure are described bellow. The main idea is to update the
rank \(R\) through density-operator-level proposals. At the density level,
with \(\boldsymbol\rho=\boldsymbol\rho(\boldsymbol\alpha,\boldsymbol\Psi)\),
the \(\varepsilon\)-stabilized cross-entropy loss gradient is
\begin{equation}
\label{eq:density-gradient}
\nabla_{\boldsymbol\rho}
\mathcal L_\varepsilon
\left(
\boldsymbol\pi(\boldsymbol\rho)
\right)
=
-
\sum_{m=1}^{M}
c_m(\boldsymbol\alpha,\boldsymbol\Psi)
\boldsymbol E_m .
\end{equation}
This gradient is used to form the negative density-level likelihood-gradient
matrix
\begin{equation}
\label{eq:negative-gradient-matrix}
\boldsymbol G(\boldsymbol\alpha,\boldsymbol\Psi)
\coloneqq
\sum_{m=1}^{M}
c_m(\boldsymbol\alpha,\boldsymbol\Psi)
\boldsymbol E_m
=
-
\nabla_{\boldsymbol\rho}
\mathcal L_\varepsilon
\left(
\boldsymbol\pi(\boldsymbol\rho)
\right).
\end{equation}
The matrix \(\boldsymbol G(\boldsymbol\alpha,\boldsymbol\Psi)\) is the
negative density-level likelihood gradient. It is used both for local atom
updates and for density-level rank adaptation as will be illustrated later.

Critical here is that the rank-adaptation and atom-update steps do not require
explicitly forming or storing the dense matrix
\(\boldsymbol G(\boldsymbol\alpha,\boldsymbol\Psi)\). Instead, they only
require products of form
\begin{equation}
\label{eq:negative-gradient-product}
\boldsymbol G(\boldsymbol\alpha,\boldsymbol\Psi)\boldsymbol v
=
\sum_{m=1}^{M}
c_m(\boldsymbol\alpha,\boldsymbol\Psi)
\boldsymbol E_m\boldsymbol v ,
\qquad
\boldsymbol v\in\mathbb C^D .
\end{equation}
Thus, when the POVM effects admit structured representations, as in Pauli,
local, sparse, or tensor-product measurements, atom responses
\eqref{eq:atom-response} and products of the form
\eqref{eq:negative-gradient-product} can be evaluated through
descriptor-level measurement actions, without storing dense POVM matrices or
dense gradient matrices. Once a density-level proposal is formed using these
matrix-free products, the active rank \(R\) is updated through the
rank-adaptation mechanism described below.

% Therefore, the optimization variables are updated centrally, while the expensive POVM-dependent products and response evaluations are distributed over outcome shards.
\vspace{-0.25 cm}
\subsection{Algorithmic setup}
\label{subsec:algorithmic-setup}

The iteration indices where workers and master iterate are indexed by $t \in \mathcal{T} := \{0,\ldots,T-1\}$. The entire time-slots $\mathcal{T}$ is partitioned into three disjoint iteration sets
\begin{equation}
\label{eq:iteration-partition}
\mathcal{T}
=
\mathcal T_{\boldsymbol{\alpha}}     \,
\cup            \,
\mathcal T_{\boldsymbol{\psi}}   \,
\cup            \,
\mathcal T_{\boldsymbol{\rho}}, 
\end{equation}
where $\mathcal{T}_{\boldsymbol{\alpha}}, \mathcal{T}_{\boldsymbol{\psi}}$, and $\mathcal{T}_{\boldsymbol{\rho}}$ represent index set of time-slots where variables $\{\boldsymbol{\psi}_r\}_{r=1}^{R}$, $\boldsymbol{\alpha}$, and $\boldsymbol{\rho}$ are updated accordingly. When \(t\in\mathcal T_{\boldsymbol{\alpha}}\), the master updates the active coefficients and may
truncate small entries, while 
\(t\in\mathcal T_{\boldsymbol{\psi}}\) workers help to update active atoms, finally when \(t\in\mathcal T_{\boldsymbol{\rho}}\), workers provide
matrix-vector products that enable memory-efficient spectral step, enabling rank-updates.

At iteration $t$, the master node have the current
rank $R_t$, the mixture weights $\boldsymbol{\alpha}_t$, and the collection of
atoms $\boldsymbol{\Psi}_t$ stored in current algorithmic state
\begin{align}
& \mathcal X_t :=
(R_t,\boldsymbol\alpha_t,\boldsymbol\Psi_t),  
\end{align}
where $\boldsymbol\alpha_t \! \in \! {\boldsymbol{\Delta}}_{R_t}, \boldsymbol\Psi_t \! =\!
\bigl[
\ket{\boldsymbol\psi_{t,1}},\ldots,
\ket{\boldsymbol\psi_{t,R_t}}
\bigr]$, and $\|\boldsymbol\psi_{t,r}\|_2 = 1$ for $r=1,\ldots,R_t$. For notational convenience, for a given
\(\mathcal X=(R,\boldsymbol\alpha,\boldsymbol\Psi)\), define
\[
\boldsymbol\rho_{\mathcal X}
\coloneqq
\boldsymbol\rho(\boldsymbol\alpha,\boldsymbol\Psi),
\qquad
\boldsymbol A_{\mathcal X}
\coloneqq
\boldsymbol A(\boldsymbol\Psi),
\]
and
\[
\boldsymbol\pi_{\mathcal X}
\coloneqq
\boldsymbol A(\boldsymbol\Psi)\boldsymbol\alpha,
\qquad
J(\mathcal X)
\coloneqq
J_R(\boldsymbol\alpha,\boldsymbol\Psi),
\]
where \(\boldsymbol\rho(\boldsymbol\alpha,\boldsymbol\Psi)\),
\(\boldsymbol A(\boldsymbol\Psi)\),
\(\boldsymbol\pi(\boldsymbol\rho)\), and
\(J_R(\boldsymbol\alpha,\boldsymbol\Psi)\) are given by
\eqref{eq:atomic-density}, \eqref{eq:atom-measurement-matrix},
\eqref{eq:atomic-born-map}, and \eqref{eq:fixed-rank-objective},
respectively. Also define
\[
\boldsymbol c_{\boldsymbol{\mathcal X}}
\coloneqq
\widehat{\boldsymbol\pi}
\oslash
\bigl(
\boldsymbol\pi_{\boldsymbol{\mathcal X}}+\varepsilon\boldsymbol 1
\bigr), \quad 
\boldsymbol G_{\boldsymbol{\mathcal X}}
\coloneqq
\sum_{m=1}^{M}
[\boldsymbol{c}_{\boldsymbol{\mathcal X}}]_m \boldsymbol E_m ,
\]
which coincide with \eqref{eq:likelihood-ratio-vector} and
\eqref{eq:negative-gradient-matrix}.

For notational brevity set $\boldsymbol\rho_t
\coloneqq
\boldsymbol\rho_{\boldsymbol{\mathcal X}_t}$, $\boldsymbol A_t
\coloneqq
\boldsymbol A_{\boldsymbol{\mathcal X}_t}$, $\boldsymbol\pi_t
\coloneqq
\boldsymbol\pi_{\boldsymbol{\mathcal X}_t}$, $J_t
\coloneqq
J(\mathcal X_t),
\boldsymbol c_t
\coloneqq
\boldsymbol c_{\boldsymbol{\mathcal X}_t},
\boldsymbol G_t
\coloneqq
\boldsymbol G_{\boldsymbol{\mathcal X}_t}$, and for
\(\boldsymbol\Psi_t=[\boldsymbol\psi_{t,1},\ldots,\boldsymbol\psi_{t,R_t}]\),
let $
\boldsymbol A_t
:=
\big[
\boldsymbol a(\boldsymbol\psi_{t,1}),
\ldots,
\boldsymbol a(\boldsymbol\psi_{t,R_t})
\big]$. This notation will be used throughout algorithmic development.

Critical to master node is to offer rank adaptation scheme. To this aim it must judiciously prune atoms with negligible weights. Specifically, after a candidate update, suppose the temporary representation contains $q$-atoms with mixture weights $\boldsymbol{\beta} \in \boldsymbol{\Delta}_q$.  Given a pruning threshold $\tau$, the master node retains the indices of those atoms whose weights exceed $\tau$ in 
\[
\mathcal I_\tau(\boldsymbol\beta)
\coloneqq
\{r\in\{1,\ldots,q\}:\beta_r>\tau\}.
\]
The corresponding mixture weights of retained indices are renormalized to form a reduced dimension probability vector 
\[
\check{\boldsymbol{\beta}} := \frac{
[\beta_r]_{r\in\mathcal I_\tau(\boldsymbol\beta)}
}{
\sum_{r\in\mathcal I_\tau(\boldsymbol\beta)}\beta_r 
}, 
 \textrm{~where~} \check{\boldsymbol{\beta}} \in \boldsymbol{\Delta}_{|\mathcal{I}_\tau(\boldsymbol{\beta})|}.
\]
The resulting representation has rank equal to $|\mathcal{I}_\tau(\boldsymbol{\beta})|$. Therefore, pruning serves as an implicit rank-reduction mechanism. Atoms with negligible contributions are removed automatically, reducing the rank without discrete rank-selection step.

Given $\mathcal{I}_{\tau}(\boldsymbol{\beta})$ and
\(\boldsymbol\Psi=[\ket{\boldsymbol\psi_1},\ldots,
\ket{\boldsymbol\psi_q}]\), let us represent the $\tau$-induced \emph{proposed} update by workers via a tuple
\begin{equation}
\label{eq:truncation-map}
\mathsf P_\tau(\boldsymbol\beta,\boldsymbol\Psi)
\coloneqq
\left(
|\mathcal I_\tau(\boldsymbol\beta)|,\,
\check{\boldsymbol{\beta}},\,
[\ket{\boldsymbol\psi_r}]_{r\in\mathcal I_\tau(\boldsymbol\beta)}
\right).
\end{equation}
The master node must accept or reject this 
proposed tuple $P_\tau(\boldsymbol\beta,\boldsymbol\Psi)$ for subsequent updates. Specifically, if a proposed tuple is accepted, master node constructs an intermediate algorithm state
\(\widetilde{\mathcal X}_{t+1}
= P_\tau(\boldsymbol\beta,\boldsymbol\Psi) = 
(\widetilde R_{t+1},
\widetilde{\boldsymbol\alpha}_{t+1},
\widetilde{\boldsymbol\Psi}_{t+1})\), and the algorithm pursues the next iteration using 
\begin{equation}
\label{eq:state-acceptance-rule}
\mathcal X_{t+1}
=
\begin{cases}
\widetilde{\mathcal X}_{t+1}, \quad \textrm{~if~}
&
J(\widetilde{\mathcal X}_{t+1})
\le
J_t-\tau_{\rm acc},
\\[0.35em]
\mathcal X_t,
&
\text{otherwise}.
\end{cases}
\end{equation}
The update rule \eqref{eq:state-acceptance-rule} ensures $\widetilde{\mathcal X}_{t+1}$ is accepted only if it gives a sufficient decrease in the objective. This update requires repeated evaluation of $
J(\mathcal{X}_t)$ and $J(\tilde{\mathcal{X}}_{t+1})$, which depend on the measurement probabilities
$
\boldsymbol{\pi}
=
\boldsymbol{A}(\boldsymbol{\Psi}_t)\boldsymbol{\alpha}_t$,
the likelihood-ratio vector
$
\boldsymbol{c}
=
\hat{\boldsymbol{\pi}}
\oslash
\left(
\boldsymbol{A}(\boldsymbol{\Psi}_t)\boldsymbol{\alpha}_t
+
\epsilon \mathbf{1}
\right)$, and the gradient-related matrix-vector products
$
\boldsymbol{A}_t^\top \boldsymbol{c}
\quad \text{and} \quad
\boldsymbol{G}_t\boldsymbol{v}$. These depend on POVM elements $\{\mathbf{E}_m\}_{m=1}^{M}$, which are too many or too large to store centrally. The distributed oracle computes these quantities without forming and storing the full measurement matrices. 
\vspace{-0.25 cm}
\subsection{Matrix-free oracle calls}
\label{sec:matrix-free-oracle}
We refer to an oracle call as a \emph{matrix-free} evaluation of
measurement-dependent quantities including
\(\boldsymbol a(\boldsymbol\psi_t)\), \(\boldsymbol A_t^\top\boldsymbol c_t\),
and \(\boldsymbol G_t\boldsymbol v\), using
POVM descriptor actions only, rather than dense matrices. To this end, let partition the $M$ measurement outcomes across $k=1,\ldots,K$ workers, where
$\mathcal{M}_k$ denote the subset of measurement indices assigned to the worker $k$. The sets $\{\mathcal{M}_k\}_{k=1}^{K}$ form a disjoint partition of
$\{1,\ldots,M\}$. Let \(M_k:=|\mathcal M_k|\).

Worker \(k\)
stores descriptors for \(\{\boldsymbol E_m:m\in\mathcal M_k\}\), along with
the corresponding slices of \(\widehat{\boldsymbol\pi}\), \(\boldsymbol\pi_t\),
and \(\boldsymbol c_t\). With
\(\boldsymbol A_{t,k}:=[\boldsymbol a_k(\boldsymbol\psi_{t,1}),\ldots,
\boldsymbol a_k(\boldsymbol\psi_{t,R_t})]\), where
\(\boldsymbol a_k(\boldsymbol\psi)\) is the restriction of
\(\boldsymbol a(\boldsymbol\psi)\) to \(\mathcal M_k\), the required oracle
quantities decompose as $
\mathcal L_\varepsilon=\sum_{k=1}^{K}\mathcal L_{\varepsilon,k}$, $
\boldsymbol A_t^\top\boldsymbol c_t
=\sum_{k=1}^{K}\boldsymbol A_{t,k}^{\top}\boldsymbol c_{t,k}$, $
\boldsymbol G_t\boldsymbol v
=\sum_{k=1}^{K}\sum_{m\in\mathcal M_k}
c_{t,m}\boldsymbol E_m\boldsymbol v$.

All measurement-dependent operations are implemented by the following
send-and-return primitives between the master node and the workers. The last
column reports the master-to-worker / worker-to-master communication size per
worker.
\begin{center}
\small
\setlength{\tabcolsep}{2.5pt}
\resizebox{\columnwidth}{!}{%
\begin{tabular}{@{}lllc@{}}
\toprule
\text{primitive} & \text{master sends} & \text{worker \(k\) returns} & \text{communication} \\
\midrule
\(\boldsymbol a(\boldsymbol\psi)\)
& \(\boldsymbol\psi\)
& \(\boldsymbol a_k(\boldsymbol\psi)\)
& \(D/M_k\) \\
\(\boldsymbol A_t^\top\boldsymbol c_t\)
& \(\boldsymbol\Psi_t,\boldsymbol c_{t,k}\)
& \(\boldsymbol A_{t,k}^{\top}\boldsymbol c_{t,k}\)
& \((DR_t+M_k)/R_t\) \\
\(\boldsymbol G_t\boldsymbol v\)
& \(\boldsymbol v,\boldsymbol c_{t,k}\)
& \(\sum_{m\in\mathcal M_k}c_{t,m}\boldsymbol E_m\boldsymbol v\)
& \((D+M_k)/D\) \\
\(\mathcal L_\varepsilon\)
& local prediction slice
& \(\mathcal L_{\varepsilon,k}\)
& \(M_k/1\) \\
\bottomrule
\end{tabular}%
}
\end{center}
The master aggregates these returned quantities and performs the projection,
scoring, truncation, and acceptance steps. The method
stores \(O(DR_t+M)\) variables globally, plus optional \(O(MR_t)\) response
caching, and never materializes dense operators \(\boldsymbol\rho_t\),
\(\{\boldsymbol E_m\}_{m=1}^{M}\), \(\mathbf A_t\), or the gradient
\(\boldsymbol G_t\), substantially reducing the storage requirements.

\vspace{-0.25 cm}
\subsection{Update rules}
\label{subsec:update-rules}

Throughout this subsection, calls to
\(\boldsymbol a(\boldsymbol\psi)\),
\(\boldsymbol A_t^\top\boldsymbol c_t\),
and \(\boldsymbol G_t\boldsymbol v\), are evaluated through the
measurement-sharded oracle described in Section~\ref{sec:matrix-free-oracle}. Thus the
optimization logic is centralized at the master, while the measurement
operations are distributed over workers leveraging the POVM shards.

For \(t\in\mathcal T_{\boldsymbol{\psi}}\), the master selects an active index
\(r_t\in\{1,\ldots,R_t\}\) to update an atom $\boldsymbol{\psi}_{r_t}$. It requests the oracle product
\(\boldsymbol G_t\ket{\boldsymbol\psi_{t,r_t}}\). Using this, it forms the atom direction
\begin{equation}
\nonumber
\label{eq:atom-direction}
\ket{\boldsymbol g_{t,r_t}}
=
\alpha_{t,r_t}
\left(
\boldsymbol I_D
-
\ket{\boldsymbol\psi_{t,r_t}}
\bra{\boldsymbol\psi_{t,r_t}}
\right)
\boldsymbol G_t
\ket{\boldsymbol\psi_{t,r_t}},
\end{equation}
and the normalized proposed atom
\begin{equation}
\label{eq:atom-trial}
\ket{\widetilde{\boldsymbol\psi}_{t+1,r_t}}
=
\frac{
\ket{\boldsymbol\psi_{t,r_t}}
+
\eta_{\boldsymbol{\psi}}\ket{\boldsymbol g_{t,r_t}}
}{
\left\|
\ket{\boldsymbol\psi_{t,r_t}}
+
\eta_{\boldsymbol{\psi}}\ket{\boldsymbol g_{t,r_t}}
\right\|_2
},
\end{equation}
using a learning rate $\eta_{\boldsymbol{\psi}}>0$.
Let \(\widetilde{\boldsymbol\Psi}_{t+1}\) be obtained from
\(\boldsymbol\Psi_t\) by replacing
\(\ket{\boldsymbol\psi_{t,r_t}}\) with
\(\ket{\widetilde{\boldsymbol\psi}_{t+1,r_t}}\), and define
\(\widetilde{\mathcal X}_{t+1}
=
(R_t,\boldsymbol\alpha_t,\widetilde{\boldsymbol\Psi}_{t+1})\).
Workers then compute the sharded response gradients

\begin{equation}
    \label{eq:pi-update}
    [\boldsymbol\pi_{\widetilde{\mathcal X}_{t+1}}
]_{\mathcal M_k}
=
[
\boldsymbol\pi_t
]_{\mathcal M_k}
+
\alpha_{t,r_t}\boldsymbol d_{t,r_t,k}.
\end{equation}

The master aggregates the values in \eqref{eq:pi-update} and and applies
\eqref{eq:state-acceptance-rule}.

For \(t\in\mathcal T_{\boldsymbol{\alpha}}\), the atoms are fixed and the master updates only
their mixture weights. Workers return
\(\boldsymbol h_{t,k}:=\boldsymbol A_{t,k}^{\top}\boldsymbol c_{t,k}\), and
the master forms
\begin{equation}
    \label{eq:h_t}
\boldsymbol h_t
:=
\sum_{k=1}^{K}\boldsymbol h_{t,k}
=
\boldsymbol A_t^\top\boldsymbol c_t.
\end{equation}
Using \eqref{eq:h_t} to form gradient in \eqref{eq:weight-gradient}, the master updates mixture weights via projected gradient descent
\begin{equation}
\label{eq:coefficient-update}
\boldsymbol\beta_t
=
\operatorname{proj}_{\boldsymbol{\Delta_{R_t}}}
\left(
\boldsymbol\alpha_t
+
\eta_{\boldsymbol{\alpha}}
\boldsymbol h_t
\right),
\end{equation}
using a learning rate $\eta_{\boldsymbol{\alpha}}>0$. 
The proposal becomes
\(\widetilde{\mathcal X}_{t+1}
=
\mathsf P_\tau(\boldsymbol\beta_t,\boldsymbol\Psi_t)\), and
\eqref{eq:state-acceptance-rule} keeps the update only if the penalized
objective decreases sufficiently.

For \(t\in\mathcal T_{\boldsymbol{\rho}}\), the algorithm performs a matrix-free spectral
refactorization step. To this aim, let us define an intermediate density operator using gradient descent with a learning rate $\eta_{\boldsymbol{\rho}}$
\begin{equation}
\label{eq:density-proposal}
\boldsymbol S_t
=
\boldsymbol\rho_t
+
\eta_{\boldsymbol{\rho}}\boldsymbol G_t .
\end{equation}
Generally, $\boldsymbol{S}_t$ is not a feasible density operator i.e., \(\boldsymbol S_t\notin\mathcal D\).
Therefore a spectral re refactorization is necessary to ensure regular projection to the feasible set. 

For spectral refactorization, we leverage the 
Ritz pairs in quantum systems. Specifically, we approximate \(\boldsymbol S_t\) using eigenpair
\((\lambda,\ket{\boldsymbol\phi})\) with \(\|\boldsymbol\phi\|_2=1\) and
\(\boldsymbol S_t\ket{\boldsymbol\phi}\approx
\lambda\ket{\boldsymbol\phi}\). In the spectral refactorization step, the
master runs a Lanczos-type routine that requests products
\(\boldsymbol S_t\ket{\boldsymbol v}\), yielding leading Ritz~pairs
\begin{equation}
\label{eq:ritz-pairs}
(\lambda_{t,\ell},\ket{\boldsymbol\phi_{t,\ell}}),
\qquad
\ell=1,\ldots,L_t,
\qquad
L_t\le R_{\max},
\end{equation}
ordered as
\(\lambda_{t,1}\ge\lambda_{t,2}\ge\cdots\ge\lambda_{t,L_t}\), with residuals
\[
\left\|
\boldsymbol S_t\ket{\boldsymbol\phi_{t,\ell}}
-
\lambda_{t,\ell}\ket{\boldsymbol\phi_{t,\ell}}
\right\|_2
\le
\epsilon_{\rm Ritz},
\]
where $L_t$ is total number of Ritz pairs.
To form the eigenpairs $\{(\lambda_{t,\ell},\ket{\boldsymbol\phi_{t,\ell}})\}_l$, the workers must return the local
pieces of \(\boldsymbol G_t\ket{\boldsymbol v}\) for any Lanczos trial vector \(\ket{\boldsymbol v}\). Having received these pieces, the master forms
\begin{equation}
\label{eq:density-product}
\boldsymbol S_t\ket{\boldsymbol v}
=
\sum_{r=1}^{R_t}
\alpha_{t,r}
\ket{\boldsymbol\psi_{t,r}}
\braket{\boldsymbol\psi_{t,r}|\boldsymbol v}
+
\eta_{\boldsymbol{\rho}}
\sum_{k=1}^{K}
\sum_{m\in\mathcal M_k}
c_{t,m}\boldsymbol E_m\ket{\boldsymbol v}.
\end{equation}
The first term uses only stored atoms, inner products, and scalar--vector
accumulations; the second term uses only descriptor-level POVM actions.
Thus the spectral step does not form dense density, POVM, or gradient
matrices.

For each $s = 1,\ldots, L_t$ master node finds
\[
\boldsymbol\Phi_{t,s}
=
[
\ket{\boldsymbol\phi_{t,1}},\ldots,\ket{\boldsymbol\phi_{t,s}}
],
\qquad
\boldsymbol\lambda_{t,s}
=
[
\lambda_{t,1},\ldots,\lambda_{t,s}
]^\top .
\]
Although \(\boldsymbol\Phi_{t,s}\) consists of normalized candidate atoms,
\(\boldsymbol\lambda_{t,s}\) is generally infeasible since
\(\boldsymbol\lambda_{t,s}\notin\Delta_s\). Hence the feasible spectral
candidates are formed by projection $\boldsymbol\beta_{t,s} =
\operatorname{proj}_{{\boldsymbol{\Delta}_s}}(\boldsymbol\lambda_{t,s})$, and $S$ candidate algorithm states are created
\begin{align}
\label{eq:spectral-candidates}
\widetilde{\mathcal X}_{t+1}^{(s)} = \mathsf P_\tau(\boldsymbol\beta_{t,s},\boldsymbol\Phi_{t,s}),
\qquad
s=1,\ldots,L_t .
\end{align}
The master node scores these candidates. Specifically,
{for a given $\mathcal X=(R,\boldsymbol\alpha,\boldsymbol\Psi)$ it scores it using rank-penalized low-rank surrogate objective
\begin{equation}
\label{eq:proximal-surrogate}
Q_t(\mathcal X)
\coloneqq
\frac{1}{2\eta_{\boldsymbol{\rho}}}
\left\|
\boldsymbol\rho_{\mathcal X}
-
\boldsymbol S_t
\right\|_F^2
+
\mu R,
\end{equation}
defined over feasible atomic representations with \(R\le R_{\max}\).} Using $Q_t(\mathcal X)$, it solves for
\begin{equation}
\label{eq:spectral-selection}
s_t^*
\in 
\underset{1\le s\le L_t}{\arg\min}
\, Q_t \bigl(\widetilde{\mathcal X}_{t+1}^{(s)}\bigr),
\end{equation}
to form a candidate $\widetilde{\mathcal X}_{t+1}
=
\widetilde{\mathcal X}_{t+1}^{(s_t^\star)}$. The acceptance test \eqref{eq:state-acceptance-rule} is finally applied to
\(\widetilde{\mathcal X}_{t+1}\) to find \({\mathcal X}_{t+1}\). This final step safeguards updated rank, and atoms. It is especially important when the Ritz pairs are approximate.

\vspace{-0.25 cm}
\subsection{Algorithm analysis}
\label{subsec:algorithm-analysis}

The next theorem characterizes the exact spectral refactorization step
induced by \eqref{eq:proximal-surrogate}--\eqref{eq:spectral-selection}. Here
``exact'' means the complete-information spectral step in which the leading
spectrum of \(\boldsymbol S_t\) is available exactly, so the spectral proposal
is not affected by Ritz approximation or numerical errors. The theorem shows
that this proposal is the global minimizer of a rank-penalized proximal
problem and yields descent of the penalized objective. The proof is deferred
to Appendix~\ref{app:algorithm-analysis-proofs}.

\begin{theorem}[Exact spectral proximal characterization and descent]
\label{thm:proximal-spectral}
Fix a spectral update iteration \(t\in\mathcal T_{\boldsymbol\rho}\). Suppose that in
\eqref{eq:ritz-pairs} the Ritz pairs are exact eigenpairs of
\(\boldsymbol S_t\), with \(L_t=R_{\max}\), and that \(\tau=0\). For each
\(s=1,\ldots,R_{\max}\), let \(\widetilde{\mathcal X}_{t+1}^{(s)}\) be the
candidate in \eqref{eq:spectral-candidates}, and let
\(\widetilde{\mathcal X}_{t+1}^{\rm prox}\) denote any proposal returned by
\eqref{eq:spectral-selection}. All minimizations below are over feasible
atomic representations \(\mathcal X=(R,\boldsymbol\alpha,\boldsymbol\Psi)\) as
in \eqref{eq:rank-adaptive-atomic-qst}. Then the following statements hold.

\begin{enumerate}[(i)]
\item For every \(s=1,\ldots,R_{\max}\),
\begin{equation}
\label{eq:rank-capped-proximal}
\widetilde{\mathcal X}_{t+1}^{(s)}
\in
\operatorname*{arg\,min}_{\mathcal X:\, R\le s}
\left\|
\boldsymbol\rho_{\mathcal X}
-
\boldsymbol S_t
\right\|_F^2 .
\end{equation}

\item The proposal \(\widetilde{\mathcal X}_{t+1}^{\rm prox}\) solves
\begin{equation}
\label{eq:exact-proximal-problem}
\widetilde{\mathcal X}_{t+1}^{\rm prox}
\in
\arg\min_{\mathcal X}
Q_t(\mathcal X).
\end{equation}

\item Assume in addition that
\(\boldsymbol\rho\mapsto\mathcal L_\varepsilon(\boldsymbol\pi(\boldsymbol\rho))\)
has \(L_\varepsilon\)-Lipschitz continuous gradient on \(\mathcal D\), and
that \(0<\eta_{\boldsymbol{\rho}}<1/ L_\varepsilon\). Then
\begin{equation}
\label{eq:proximal-descent}
J\!\left(\widetilde{\mathcal X}_{t+1}^{\rm prox}\right)
\le
J_t
-
\left(
\frac{1}{2\eta_{{\boldsymbol{\rho}}}}
-
\frac{L_\varepsilon}{2}
\right)
\left\|
\boldsymbol\rho_{\widetilde{\mathcal X}_{t+1}^{\rm prox}}
-
\boldsymbol\rho_t
\right\|_F^2 .
\end{equation}
In particular,
\(J(\widetilde{\mathcal X}_{t+1}^{\rm prox})\le J_t\).
\end{enumerate}
\end{theorem}

Theorem~\ref{thm:proximal-spectral} suggests that the spectral refactorization
step is not an ad-hoc truncation. Having the exact eigenpairs, the candidate family
\eqref{eq:spectral-candidates} contains the exact rank-penalized proximal
update, and \eqref{eq:spectral-selection} recovers it. In the implemented
algorithm, exact eigenpairs are replaced by matrix-free Ritz pairs computed
through \eqref{eq:density-product}, while
\eqref{eq:state-acceptance-rule} remains a practical safeguard under these
approximations.

\begin{algorithm}[t]
\footnotesize
\caption{Master--worker matrix-free rank-adaptive QST}
\label{alg:matrix-free-qst}
\begin{algorithmic}[1]
\REQUIRE \(\widehat{\boldsymbol\pi}\), POVM shards
\(\{\mathcal M_k\}_{k=1}^{K}\), initial
\(\mathcal X_0=(R_0,\boldsymbol\alpha_0,\boldsymbol\Psi_0)\),
sets \(\mathcal T_{\boldsymbol\psi},\mathcal T_{\boldsymbol{\alpha}},\mathcal T_{\boldsymbol{\rho}}\),
and parameters \(\eta_\psi,\eta_\alpha,\eta_\rho,\mu,\varepsilon,\tau,
\tau_{\rm acc},R_{\max},T\)
\STATE Master initializes
\(\boldsymbol\pi_0=\boldsymbol A(\boldsymbol\Psi_0) \, \boldsymbol\alpha_0\),
\(J_0=J(\mathcal X_0)\)
\FOR{\(t=0,\ldots,T-1\)}
    \STATE Master forms
    \(\boldsymbol c_t\) \eqref{eq:likelihood-ratio-vector}, and sends slices
    \(\boldsymbol c_{t,k}\) back to workers 
    \IF{\(t\in\mathcal T_{\boldsymbol\psi}\)}
        \STATE Master selects \(r_t\); workers return pieces
        \(\boldsymbol G_t\ket{\boldsymbol\psi_{t,r_t}}\)
        \STATE Master forms \(\widetilde{\mathcal X}_{t+1}\)
        \eqref{eq:atom-direction}--\eqref{eq:atom-trial}; workers return $\boldsymbol{d}_{r_t,k}$
    \ELSIF{\(t\in\mathcal T_{\boldsymbol{\alpha}}\)}
        \STATE Workers return
        \(\boldsymbol A_{t,k}^{\top}\boldsymbol c_{t,k}\); master forms
        \(\boldsymbol A_t^\top\boldsymbol c_t\), updates
        \(\boldsymbol\beta_t\) by \eqref{eq:coefficient-update} and sets
        \(\widetilde{\mathcal X}_{t+1}
        =\mathsf P_\tau(\boldsymbol\beta_t,\boldsymbol\Psi_t)\)
    \ELSIF{\(t\in\mathcal T_{\boldsymbol{\rho}}\)}
        \STATE Master runs Lanczos; workers provide the oracle products in
        \eqref{eq:density-product} Master forms Ritz candidates by
        \eqref{eq:spectral-candidates} and selects
        \(\widetilde{\mathcal X}_{t+1}\) by
        \eqref{eq:spectral-selection}
    \ENDIF
    \STATE Workers return values in \eqref{eq:pi-update} \\ Master evaluates
    \(J(\widetilde{\mathcal X}_{t+1})\) and applies
    \eqref{eq:state-acceptance-rule} 
\ENDFOR
\STATE \textbf{return} \((R_T,\boldsymbol\alpha_T,\boldsymbol\Psi_T)\)
\end{algorithmic}
\end{algorithm}

% \vspace{-0.35 cm}
\section{Numerical tests}
\label{sec:simulations}

We evaluate Algorithm~\ref{alg:matrix-free-qst} on \(N\)-qubit
low-rank QST instances. The simulation parameters are the
number of qubits \(N\), the true rank \(R_{\star}\), the number of measured
Pauli strings \(B\), the number of shots per Pauli string \(N_s\), and the rank budget \(R_{\max}\). The derived quantities are
\(D=2^N\), \(M=BD\), and total shots \(N_{\rm shot}=BN_s\). The ground-truth state follows the atomic model
\[
\boldsymbol\rho_{\star}
=
\sum_{r=1}^{R_{\star}}
\alpha_{\star,r}
\ket{\boldsymbol\psi_{\star,r}}
\bra{\boldsymbol\psi_{\star,r}} .
\]
The weights \(\boldsymbol\alpha_{\star}\in\Delta_{R_{\star}}\) are drawn from
a Dirichlet distribution with unit concentration parameters and sorted in
nonincreasing order. The atoms are generated independently as
\[
\boldsymbol\psi_{\star,r}
=
\boldsymbol\xi_r/\|\boldsymbol\xi_r\|_2, 
\quad \textrm{~where~}
\boldsymbol\xi_r\sim\mathcal{CN}(\boldsymbol 0,\boldsymbol I_D),
\,
r=1,\ldots,R_{\star}.
\]

The POVM design consists of \(B\) random Pauli strings. A Pauli string
specifies one single-qubit measurement, \(X\), \(Y\), or \(Z\), on each of the
\(N\) qubits. Thus, the \(b\)-th string is
\[
\boldsymbol s_b
=
(s_{b,1},\ldots,s_{b,N})
\in
\{X,Y,Z\}^{N},
\qquad
b=1,\ldots,B
\]
drawn uniformly at random. Each string defines a product measurement basis
with \(D=2^N\) possible outcomes indexed by
\(\boldsymbol u\in\{0,1\}^{N}\). Let \(\boldsymbol V_b\) denote the product
basis rotation associated with \(\boldsymbol s_b\), i.e.,
\(\boldsymbol V_b=\boldsymbol V_{s_{b,1}}\otimes\cdots\otimes
\boldsymbol V_{s_{b,N}}\), where \(\boldsymbol V_X\), \(\boldsymbol V_Y\),
and \(\boldsymbol V_Z\) rotate the corresponding Pauli eigenbasis to the
computational basis.

To express these measurements using the POVM notation of
Section~\ref{sec:formulation}, we identify a global outcome with the pair
\((b,\boldsymbol u)\) and define
\[
\boldsymbol E_{b,\boldsymbol u}
=
\frac{1}{B}
\boldsymbol V_b^{*}
\ket{\boldsymbol u}\bra{\boldsymbol u}
\boldsymbol V_b .
\]
The factor \(1/B\) combines the \(B\) separately chosen Pauli strings into one
POVM with $
\sum_{b=1}^{B}
\sum_{\boldsymbol u\in\{0,1\}^{N}}
\boldsymbol E_{b,\boldsymbol u}
=
\boldsymbol I_D.$

Dense matrices \(\boldsymbol E_{b,\boldsymbol u}\) are used only for
notation; in the implementation each measurement is represented by its Pauli
string \(\boldsymbol s_b\).

For data generation, the physical outcome probabilities are conditional on
the chosen Pauli string. For string \(b\), we have
\[
\mathbb P(\boldsymbol u\mid b,\boldsymbol\rho_{\star})
=
\bra{\boldsymbol u}
\boldsymbol V_b
\boldsymbol\rho_{\star}
\boldsymbol V_b^{*}
\ket{\boldsymbol u}
=
B\operatorname{tr}
\left(
\boldsymbol E_{b,\boldsymbol u}\boldsymbol\rho_{\star}
\right).
\]
We draw \(N_s\) shots for each Pauli string. If
\(N_{b,\boldsymbol u}\) is the count of outcome \(\boldsymbol u\) under
string \(b\), then the empirical frequency supplied to the algorithm is
$
\widehat\pi_{b,\boldsymbol u}
=
\frac{N_{b,\boldsymbol u}}{BN_s},$ and
$\widehat{\boldsymbol\pi}\in\Delta_M$. Thus, the data consist of \(B\) empirical histograms, each with \(D\) outcomes,
stacked into the single vector \(\widehat{\boldsymbol\pi}\) used in
\eqref{eq:qst-loss}.

Table~\ref{tab:sim-parameters} summarizes the default simulation parameters.
The learning rate is obtained through grid search over \(\mathcal G_{\eta}\) using a simple backtracking; that is, during each slot, per
proposed atom, coefficient, or refactorization update, candidates from
$\mathcal{G}_\eta$ are tested from largest to smallest, and the largest accepted is used.

\begin{table}[t]
\centering
\caption{Default simulation parameters.}
\label{tab:sim-parameters}
\footnotesize
\setlength{\tabcolsep}{3pt}
\begin{tabularx}{\columnwidth}{@{}lX@{}}
\toprule
Quantity & Value \\
\midrule
Qubits & \(N\), with \(D=2^N\) \\
True rank & \(R_{\star}=4\) \\
Pauli strings & \(B=16N\), unless varied \\
Shots per string & \(N_s=512\) \\
Derived size & \(M=BD,\quad N_{\rm shot}=BN_s\) \\
Rank budget / init. &
\(R_{\max}=8,\quad R_0=R_{\max},\quad
\boldsymbol\alpha_0=R_{\max}^{-1}\boldsymbol 1\) \\
Objective params. & \(\varepsilon=10^{-8},\quad \mu=10^{-3}\) \\
Acceptance / pruning & \(\tau_{\rm acc}=10^{-10},\quad \tau=10^{-3}\) \\
Schedules & \(T_{\alpha}=5,\quad T_{\rm ref}=50\) \\
Ritz step & \(L_t=R_{\max},\quad N_{\rm Lanczos}\le30\) \\
Step grid & \(\mathcal G_{\eta}=\{2^{-j}\}_{j=0}^{10}\) \\
Stopping & \(\delta_{20}<10^{-4}\), or \(N_{\rm oracle}=10^5\), or
\(m_{\rm peak}=64\) GiB \\
Trials & 20 independent instances; medians reported \\
\bottomrule
\end{tabularx}
\vspace{0.25em}
\begin{flushleft}
\footnotesize
Here \(\delta_{20}\) is the relative objective decrease over the last 20
accepted updates, \(N_{\rm oracle}\) is the serial-equivalent oracle-call
count, and \(m_{\rm peak}\) is peak memory.
\end{flushleft}
\end{table}

The compared methods are listed in Table~\ref{tab:sim-baselines}. All methods
use the same Pauli strings, empirical vector \(\widehat{\boldsymbol\pi}\),
stabilized likelihood \eqref{eq:qst-loss}, stopping criteria, and step grid
\(\mathcal G_{\eta}\). Fixed-rank baselines are given the true rank
\(R_{\star}\), while the proposed method is given only \(R_{\max}\). Dense
matrices are formed only when computing the final normalized  square error
$
e_{\rho}
:=
\frac{
\|\widehat{\boldsymbol\rho}-\boldsymbol\rho_{\star}\|_F
}{
\|\boldsymbol\rho_{\star}\|_F
}$.

\begin{table}[t]
\centering
\caption{Baseline markers used in the simulation figures.}
\label{tab:sim-baselines}
\footnotesize
\setlength{\tabcolsep}{2.5pt}
\begin{tabularx}{\columnwidth}{@{}lX@{}}
\toprule
Marker / method & Implementation \\
\midrule
\qstDenseMark{} \textsc{Dense PGD/MLE}
& Full-rank projected-gradient MLE over \(\mathcal D\)
\cite{Bolduc2017PGDQST}.
\\
\qstLRMark{} \textsc{LR-PGD}
& Rank-\(R_{\star}\) projected-gradient QST
\cite{Gross2010CSQST,Flammia2012CSQST,Bolduc2017PGDQST}.
\\
\qstRGDMark{} \textsc{Factored RGD}
& Rank-\(R_{\star}\) factored/Riemannian QST
\cite{Kyrillidis2018NonconvexQST,Hsu2024RGDQST,Wang2024EfficientFGD}.
\\
\qstFWMark{} \textsc{Atomic FW}
& Rank-one Frank--Wolfe with corrective reweighting
\cite{Hazan2008SparseSDP,Jaggi2013FrankWolfe}.
\\
\qstPropMark{} \textsc{Our Proposed}
& Rank-adaptive matrix-free atomic QST.
\\
\bottomrule
\end{tabularx}
\end{table}

\vspace{-0.15 cm}
\subsection{Results}
\label{subsec:simulation_results}

\begin{figure}[t]
\centering

\subfloat[Statistical recovery\label{fig:qst_comparison_a}]{%
  \includegraphics[width=0.48\columnwidth]{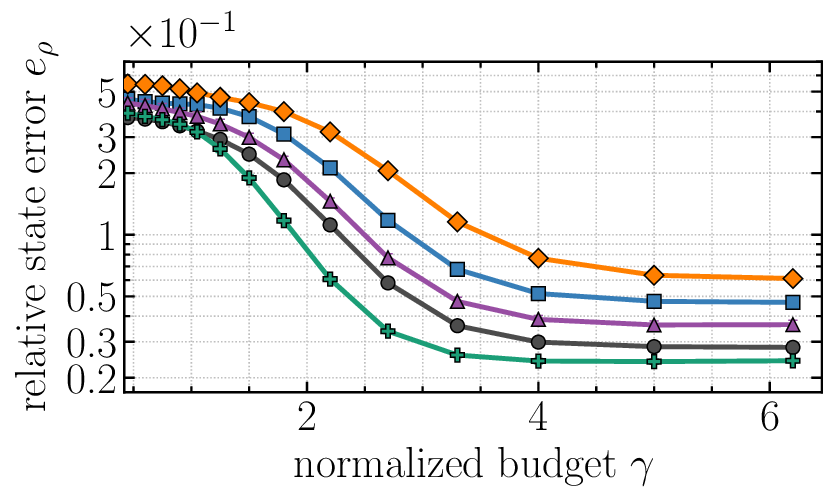}
}
\hfill
\subfloat[Accuracy--runtime tradeoff\label{fig:qst_comparison_b}]{%
  \includegraphics[width=0.48\columnwidth]{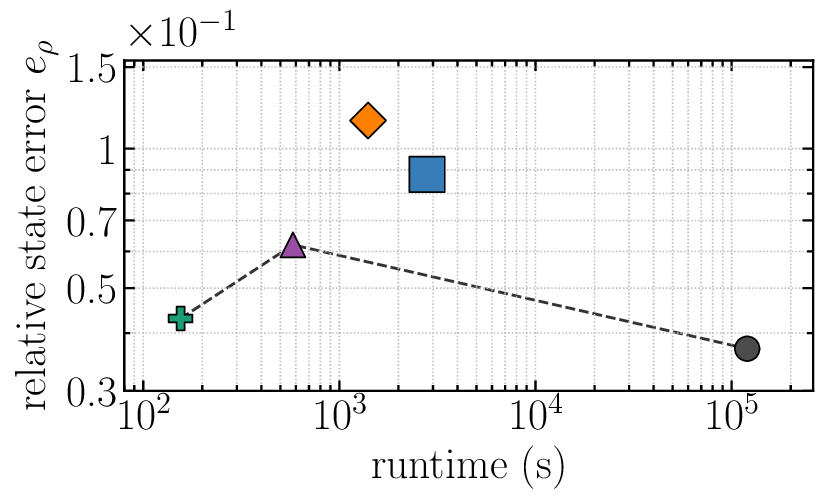}
}

\vspace{0.5em}
\subfloat[Runtime scaling\label{fig:qst_comparison_c}]{%
  \includegraphics[width=0.48\columnwidth]{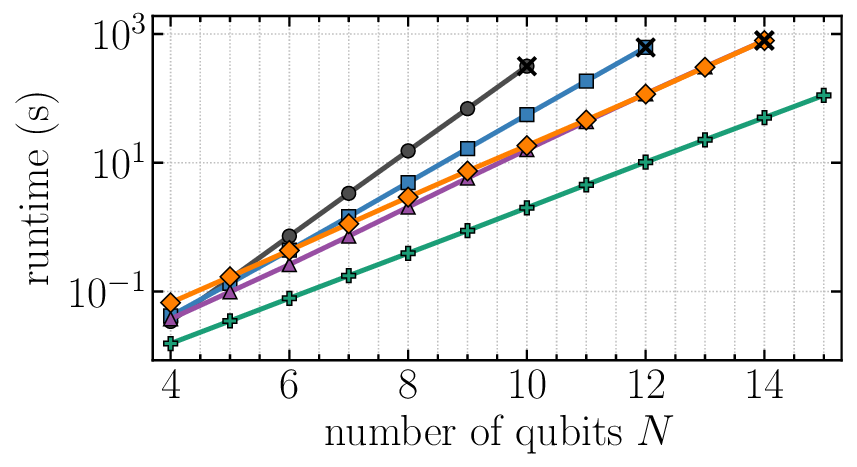}
}
\hfill
\subfloat[Memory scaling\label{fig:qst_comparison_d}]{%
  \includegraphics[width=0.48\columnwidth]{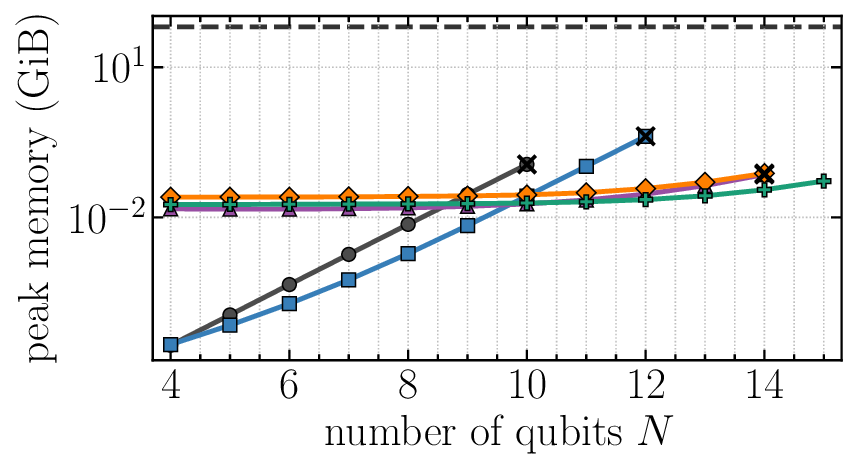}
}

\caption{
Comparison with representative QST solvers. Method markers follow
Table~\ref{tab:sim-baselines}. Fig.~\ref{fig:qst_comparison_a} uses
\(N=12\) and \(R_{\star}=4\), while varying \(\gamma\).
Fig.~\ref{fig:qst_comparison_b} reports accuracy versus runtime at \(N=14\)
and \(\gamma=2\), with marker area proportional to peak memory.
Figs.~\ref{fig:qst_comparison_c}--\ref{fig:qst_comparison_d} show runtime
and memory scaling with \(N\) at \(\gamma=2\). The marker \qstTermMark{}
indicates termination by stopping rule or resource cap, and \qstCapLine{}
marks the 64 GiB memory cap.
}
% \vspace{-0.25 cm}
\label{fig:qst_comparison}
\end{figure}

Fig.~\ref{fig:qst_comparison} benchmarks the performance of our proposed method against alternatives in terms of estimation error, memory-efficiency, and scalability. In Fig.~\ref{fig:qst_comparison_a}, increasing normalized budget 
\(\gamma\) (where $B \propto \gamma \times N$) improves estimation error by increasing both the number of Pauli settings
and the total shot count. At small budgets, sampling error dominates and the
methods are close, whereas at larger budgets the benefit of low-rank modeling
becomes pronounced. Dense PGD/MLE estimates an unrestricted \(D\times D\)
density matrix, while the low-rank methods search over smaller models. Our
proposed method additionally penalizes the rank order, helping remove
unnecessary atoms as the measurements become more informative.

Fig.~\ref{fig:qst_comparison_b} shows the accuracy--runtime--memory tradeoff. The larger the marker, the bigger the memory requirement of the algorithm. 
Dense PGD/MLE method requires full spectral projections, LR-PGD reduces the retained
rank but still uses matrix-level projected steps, factored RGD avoids
full-rank variables but requires the prescribed factor rank \(R_{\star}\), and
atomic-FW is matrix-free but changes rank through via greedy updates. Our novel approach
combines atom updates
\eqref{eq:atom-direction}--\eqref{eq:atom-trial}, simplex reweighting
\eqref{eq:coefficient-update}, and spectral refactorization
\eqref{eq:density-proposal}--\eqref{eq:spectral-selection}, yielding a
favorable accuracy--runtime tradeoff without dense projections --- as it is closer to the left bottom corner of the plot. In addition, it has a smaller memory requirement, signified by the size of its marker. 

Figs.~\ref{fig:qst_comparison_c} and \ref{fig:qst_comparison_d} illustrate the
runtime--memory advantage of matrix-free oracle evaluations. Dense methods
require \(O(D^2)\) memory and dense spectral operations, which become
increasingly limiting as \(N\) grows. Our novel  method stores active atoms,
coefficients, probability vectors, and Lanczos work vectors, with working
memory \(O(DR_t+M)\) up to optional response storage. Pauli effects and
likelihood-gradient actions are evaluated through descriptors, so dense
density, POVM, and gradient matrices are never formed.

For the algorithmic ablation, we compare four variants,
$
\mathcal V
:=
\{
v_{\rm atom},
v_{\rm coeff},
v_{\rm kr},
v_{\rm full}
\}$.
All variants use the same data, stabilized likelihood, initialization, step
grid, and acceptance rule; they differ only in which candidate moves are
enabled. The variant \(v_{\rm atom}\), denoted \textsc{Atom only}, enables
only the local atom trial
\eqref{eq:atom-direction}--\eqref{eq:atom-trial}; the coefficients and active
support size are fixed. The variant \(v_{\rm coeff}\), denoted
\textsc{Atom + coeff.}, also enables the simplex coefficient update
\eqref{eq:coefficient-update}; the current atoms can therefore be moved and
reweighted, but no spectral refactorization is used. The variant
\(v_{\rm kr}\), denoted \textsc{Known-rank refactor.}, further enables the
spectral refactorization proposal
\eqref{eq:density-proposal}--\eqref{eq:spectral-selection}, but restricts the
refactorized support to the true rank \(R_{\star}\). The variant
\(v_{\rm full}\), denoted \textsc{Full adaptive}, is the complete proposed
method: it uses spectral refactorization and pruning, and the active rank
\(R_t\) is selected by the rank-penalized objective rather than supplied as
input. For variant \(v\in\mathcal V\), let \(J_{v,i}^{\rm rec}\) be the recorded
penalized objective value after its \(i\)th accepted update, with
\(i=0,\ldots,I_v\). We plot
\[
J_{\rm best}
:=
\min_{v\in\mathcal V}
\min_{0\le i\le I_v}
J_{v,i}^{\rm rec},
\qquad
\Delta J_v(i)
:=
J_{v,i}^{\rm rec}-J_{\rm best}.
\]
Thus, \(\Delta J_v(i)\) is the objective gap of variant \(v\) relative to the
best objective value attained by any displayed variant on the same instance.

\begin{figure*}[t]
\centering
\subfloat[Objective gap\label{fig:alg_ablation_obj}]{%
  \includegraphics[width=0.305\textwidth]{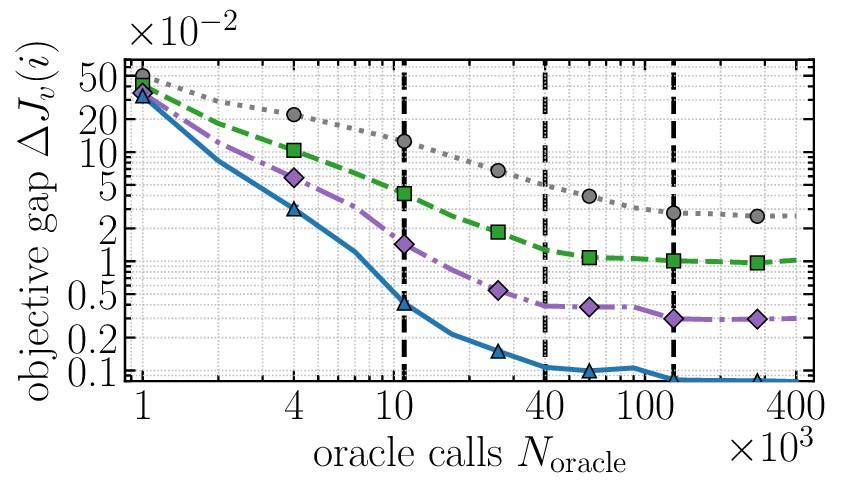}
}
\hspace{0.018\textwidth}
\subfloat[Recovery error\label{fig:alg_ablation_err}]{%
  \includegraphics[width=0.305\textwidth]{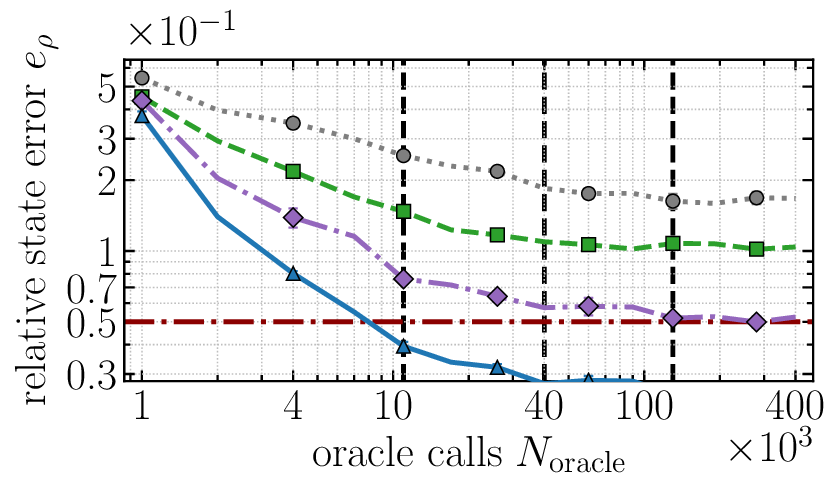}
}
\hspace{0.018\textwidth}
\subfloat[Active rank\label{fig:alg_ablation_rank}]{%
  \includegraphics[width=0.305\textwidth]{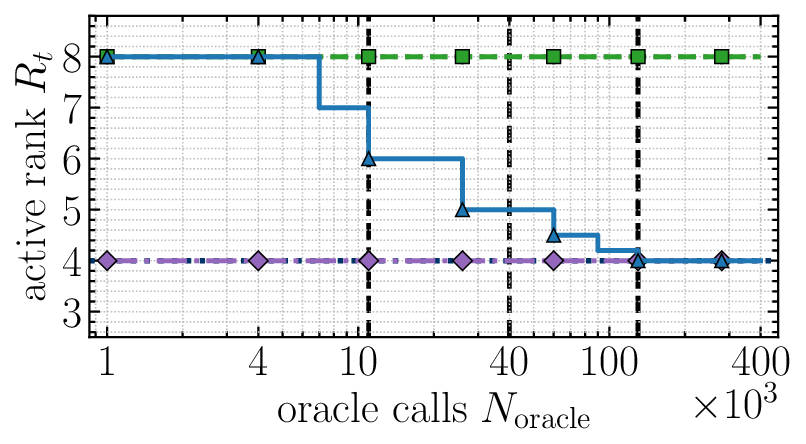}
}
\caption{
Algorithmic ablation and rank adaptation for the variants
\(v_{\rm atom}\), \(v_{\rm coeff}\), \(v_{\rm kr}\), and \(v_{\rm full}\). Markers indicate
\algAtomMark{} \textsc{Atom only},
\algCoeffMark{} \textsc{Atom + coeff.},
\algOracleMark{} \textsc{Known-rank refactor.},
and \algFullMark{} \textsc{Full adaptive}.
The horizontal axis counts serial-equivalent matrix-free oracle calls.
Fig.~\ref{fig:alg_ablation_obj} shows \(\Delta J_v(i)\);
Fig.~\ref{fig:alg_ablation_err} shows \(e_{\rho}\), with \algErrorLine{}
denoting the 5\% level; and Fig.~\ref{fig:alg_ablation_rank} shows \(R_t\),
with \algKstarLine{} denoting \(R_{\star}=4\). The dashed line
\algRefactorLine{} marks spectral refactorization epochs.
}
% \vspace{-0.5 cm}
\label{fig:alg_ablation_rank_adaptation}
\end{figure*}
Fig.~\ref{fig:alg_ablation_rank_adaptation} analyzes the contribution of each
candidate move through objective decrease, recovery accuracy, and active-rank
evolution. In Fig.~\ref{fig:alg_ablation_obj}, local atom updates reduce the
objective initially but can stall when the current rank is poorly placed,
while adding the simplex update improves reweighting of the existing atoms.
Spectral refactorization produces judicious rank updates, and the
known-rank variant shows the behavior when 
\(R_{\star}\) is
provided. Fig.~\ref{fig:alg_ablation_err} shows that these rank updates
also improve recovery error. Finally, Fig.~\ref{fig:alg_ablation_rank} shows
that the full adaptive method can successfully unravel true \(R_{\star}\), the truncation map
\eqref{eq:truncation-map} removes numerically inactive coefficients, and
the rank penalty discourages retaining extra Ritz atoms unless they improve
the penalized objective.

%%%%%%%%%%%%%%%%%%%%%%%%%%%%%%%%%%%%%%%%%%%%%%%%%%%%%%%%%%%%%%%%%%%%%%%%%%%%%%%%%%%%%%%%%%%%%%%%%%%%%%%%%%%%%%%
\section{Concluding Summary}
\label{sec:conclusion}

This paper developed a rank-adaptive matrix-free method for low-rank QST. Representing the density operator as a convex combination of rank-one atoms, the method maintains a physically feasible
estimate while avoiding dense density, measurement, and gradient matrices.
The algorithm combined atom updates, weight updates, and spectral
refactorization under a rank penalty, which allowed the rank to adapt during
optimization. Its measurement computations were distributed across
workers without forming large matrices. Feasibility,
prediction consistency, monotone descent were established, along with an exact spectral proximal
characterization of the refactorization step. Numerical tests demonstrated favorable accuracy--runtime--memory tradeoffs relative to representative baselines.

\appendix
\section{Algorithm Analysis Proofs}
\label{app:algorithm-analysis-proofs}

\begin{proof}[Proof of Theorem~\ref{thm:proximal-spectral}]
Fix \(t\in\mathcal T_{\boldsymbol\rho}\), and let
\[
\boldsymbol S_t
=
\boldsymbol U_t
\operatorname{diag}(\lambda_{t,1},\ldots,\lambda_{t,D})
\boldsymbol U_t^*,
\qquad
\lambda_{t,1}\ge\cdots\ge\lambda_{t,D}.
\]
Under the exact-spectrum assumptions, the Ritz vectors in
\(\boldsymbol\Phi_{t,s}\) are the leading \(s\) eigenvectors of
\(\boldsymbol S_t\), and
\[
\boldsymbol\beta_{t,s}
=
\operatorname{proj}_{\boldsymbol{\Delta}_s}
(\boldsymbol\lambda_{t,s}),
\qquad
\boldsymbol\lambda_{t,s}
=
[\lambda_{t,1},\ldots,\lambda_{t,s}]^\top .
\]
Since \(\tau=0\), the truncation map removes only zero coefficients, and does
not change the represented density; hence,
\[
\boldsymbol\rho_{\widetilde{\mathcal X}_{t+1}^{(s)}}
=
\boldsymbol\Phi_{t,s}
\operatorname{diag}(\boldsymbol\beta_{t,s})
\boldsymbol\Phi_{t,s}^{*}.
\]
To prove (i), consider any
\(\boldsymbol\rho\in\mathcal D\) with
\(\operatorname{rank}(\boldsymbol\rho)\le s\), and write
\[
\begin{aligned}
\boldsymbol\rho
&=
\boldsymbol V
\operatorname{diag}(\gamma_1,\ldots,\gamma_D)
\boldsymbol V^*,\\
&\qquad
\gamma_1\ge\cdots\ge\gamma_D\ge0,
\quad
\sum_{i=1}^{D}\gamma_i=1,
\quad
\gamma_i=0 \ \text{for } i>s .
\end{aligned}
\]
It then holds that 
\[
\|\boldsymbol\rho-\boldsymbol S_t\|_F^2
=
\|\boldsymbol\gamma\|_2^2
+
\|\boldsymbol\lambda_t\|_2^2
-
2\operatorname{tr}(\boldsymbol\rho\boldsymbol S_t).
\]
For Hermitian matrices, von Neumann's trace inequality implies
\[
\operatorname{tr}(\boldsymbol\rho\boldsymbol S_t)
\le
\sum_{i=1}^{D}\gamma_i\lambda_{t,i}
\]
with equality when \(\boldsymbol\rho\) and \(\boldsymbol S_t\) are
simultaneously diagonalized with eigenvalues ordered in the same
nonincreasing order.~Thus,
\[
\|\boldsymbol\rho-\boldsymbol S_t\|_F^2
\ge
\sum_{i=1}^{D}(\gamma_i-\lambda_{t,i})^2 ,
\]
and equality is attained by choosing the eigenvectors of
\(\boldsymbol\rho\) to be the leading eigenvectors of \(\boldsymbol S_t\).
Since \(\gamma_i=0\) for \(i>s\), minimizing over all feasible
rank-\(s\) densities is equivalent to
\[
\min_{\boldsymbol\beta\in\boldsymbol{\Delta}_s}
\left\|
\boldsymbol\beta-\boldsymbol\lambda_{t,s}
\right\|_2^2 
\]
up to the constant
\(\sum_{i>s}\lambda_{t,i}^2\). Its minimizer is $ \boldsymbol\beta_{t,s}
=
\operatorname{proj}_{\boldsymbol{\Delta}_s}
(\boldsymbol\lambda_{t,s})
$
Thus,
\(\boldsymbol\rho_{\widetilde{\mathcal X}_{t+1}^{(s)}}\) minimizes
\(\|\boldsymbol\rho-\boldsymbol S_t\|_F^2\) over all
\(\boldsymbol\rho\in\mathcal D\) with rank at most \(s\). Since the objective
in \eqref{eq:rank-capped-proximal} depends on \(\mathcal X\) only through
\(\boldsymbol\rho_{\mathcal X}\), this proves (i).

For (ii), let \(\mathcal X\) be any feasible representation with active size
\(R\). By (i), it follows that
\[
\left\|
\boldsymbol\rho_{\widetilde{\mathcal X}_{t+1}^{(R)}}
-
\boldsymbol S_t
\right\|_F^2
\le
\left\|
\boldsymbol\rho_{\mathcal X}
-
\boldsymbol S_t
\right\|_F^2 .
\]
Moreover, since \(\tau=0\), the truncation map can only remove zero
coefficients, so the active size of
\(\widetilde{\mathcal X}_{t+1}^{(R)}\) is at most \(R\); hence,
\[
Q_t\!\left(\widetilde{\mathcal X}_{t+1}^{(R)}\right)
\le
Q_t(\mathcal X).
\]
Thus, a global minimizer of \(Q_t\) is attained by at least one member of the
candidate family
\(\{\widetilde{\mathcal X}_{t+1}^{(s)}\}_{s=1}^{R_{\max}}\). As a result, 
\eqref{eq:spectral-selection} returns a solution of
\eqref{eq:exact-proximal-problem}, which proves (ii).

As far as (iii) is concerned, define the density displacement
\[
\boldsymbol H_t
\coloneqq
\boldsymbol\rho_{\widetilde{\mathcal X}_{t+1}^{\rm prox}}
-
\boldsymbol\rho_t .
\]
Since
\(\widetilde{\mathcal X}_{t+1}^{\rm prox}\) minimizes \(Q_t\) and
\(\mathcal X_t\) is feasible, it follows that
\[
Q_t\!\left(\widetilde{\mathcal X}_{t+1}^{\rm prox}\right)
\le
Q_t(\mathcal X_t).
\]
Using
\(\boldsymbol S_t=\boldsymbol\rho_t+\eta_\rho\boldsymbol G_t\) and 
\[
\boldsymbol G_t
=
-
\nabla_{\boldsymbol\rho}
\mathcal L_\varepsilon
\bigl(
\boldsymbol\pi(\boldsymbol\rho_t)
\bigr),
\]
we obtain
\[
\left\langle
\nabla_{\boldsymbol\rho}
\mathcal L_\varepsilon
\bigl(
\boldsymbol\pi(\boldsymbol\rho_t)
\bigr),
\boldsymbol H_t
\right\rangle_F
+
\frac{1}{2\eta_\rho}
\|\boldsymbol H_t\|_F^2
+
\mu R\!\left(\widetilde{\mathcal X}_{t+1}^{\rm prox}\right)
\le
\mu R_t .
\]
By \(L_\varepsilon\)-smoothness on \(\mathcal D\), we have
\[
\begin{aligned}
\mathcal L_\varepsilon
\!\left(
\boldsymbol\pi
\bigl(
\boldsymbol\rho_{\widetilde{\mathcal X}_{t+1}^{\rm prox}}
\bigr)
\right)
&\le
\mathcal L_\varepsilon
\bigl(
\boldsymbol\pi(\boldsymbol\rho_t)
\bigr) \\
&\quad
+
\left\langle
\nabla_{\boldsymbol\rho}
\mathcal L_\varepsilon
\bigl(
\boldsymbol\pi(\boldsymbol\rho_t)
\bigr),
\boldsymbol H_t
\right\rangle_F
+
\frac{L_\varepsilon}{2}
\|\boldsymbol H_t\|_F^2 .
\end{aligned}
\]

Adding
\(\mu R(\widetilde{\mathcal X}_{t+1}^{\rm prox})\) on both sides and using the last inequality
\[
\begin{aligned}
J\!\left(\widetilde{\mathcal X}_{t+1}^{\rm prox}\right)
&\le
J_t
-
\left(
\frac{1}{2\eta_\rho}
-
\frac{L_\varepsilon}{2}
\right)
\|\boldsymbol H_t\|_F^2 \\
&=
J_t
-
\left(
\frac{1}{2\eta_\rho}
-
\frac{L_\varepsilon}{2}
\right)
\left\|
\boldsymbol\rho_{\widetilde{\mathcal X}_{t+1}^{\rm prox}}
-
\boldsymbol\rho_t
\right\|_F^2 .
\end{aligned}
\]
Since \(0<\eta_\rho<1/L_\varepsilon\), the coefficient is positive, and
\(J(\widetilde{\mathcal X}_{t+1}^{\rm prox})\le J_t\), which concludes the proof of
\eqref{eq:proximal-descent}.
\end{proof}

{\scriptsize
\bibliographystyle{IEEEtran}
\bibliography{references}

@article{Gross2010CSQST,
  title   = {Quantum State Tomography via Compressed Sensing},
  author  = {Gross, David and Liu, Yi-Kai and Flammia, Steven T. and Becker, Stephen and Eisert, Jens},
  journal = {Phys. Rev. Lett.},
  volume  = {105},
  number  = {15},
  pages   = {150401},
  year    = {2010},
  doi     = {10.1103/PhysRevLett.105.150401}
}

@article{Flammia2012CSQST,
  title   = {Quantum Tomography via Compressed Sensing: Error Bounds, Sample Complexity and Efficient Estimators},
  author  = {Flammia, Steven T. and Gross, David and Liu, Yi-Kai and Eisert, Jens},
  journal = {New J. Phys.},
  volume  = {14},
  number  = {9},
  pages   = {095022},
  year    = {2012},
  doi     = {10.1088/1367-2630/14/9/095022}
}

@misc{Haah2017SampleOptimalQST,
  title         = {Optimal Lower Bounds for Quantum State Tomography},
  author        = {Scharnhorst, Thilo and Spilecki, Jack and Wright, John},
  howpublished  = {arXiv:2510.07699},
  year          = {2025},
  eprint        = {2510.07699},
  archivePrefix = {arXiv},
  primaryClass  = {quant-ph},
  doi           = {10.48550/arXiv.2510.07699}
}

@article{Bolduc2017PGDQST,
  title   = {Projected Gradient Descent Algorithms for Quantum State Tomography},
  author  = {Bolduc, Eliot and Knee, George C. and Gauger, Erik M. and Leach, Jonathan},
  journal = {npj Quantum Inf.},
  volume  = {3},
  pages   = {44},
  year    = {2017},
  doi     = {10.1038/s41534-017-0043-1}
}

@article{Kyrillidis2018NonconvexQST,
  title   = {Provable Compressed Sensing Quantum State Tomography via Non-Convex Methods},
  author  = {Kyrillidis, Anastasios and Kalev, Amir and Park, Dohyung and Bhojanapalli, Srinadh and Caramanis, Constantine and Sanghavi, Sujay},
  journal = {npj Quantum Inf.},
  volume  = {4},
  pages   = {36},
  year    = {2018},
  doi     = {10.1038/s41534-018-0080-4}
}

@article{Hsu2024RGDQST,
  title   = {Quantum State Tomography via Nonconvex {Riemannian} Gradient Descent},
  author  = {Hsu, Ming-Chien and Kuo, En-Jui and Yu, Wei-Hsuan and Cai, Jian-Feng and Hsieh, Min-Hsiu},
  journal = {Phys. Rev. Lett.},
  volume  = {132},
  number  = {24},
  pages   = {240804},
  year    = {2024},
  doi     = {10.1103/PhysRevLett.132.240804}
}

@article{Wang2024EfficientFGD,
  title   = {Efficient Factored Gradient Descent Algorithm for Quantum State Tomography},
  author  = {Wang, Yong and Liu, Lijun and Cheng, Shuming and Li, Li and Chen, Jie},
  journal = {Phys. Rev. Res.},
  volume  = {6},
  number  = {3},
  pages   = {033034},
  year    = {2024},
  doi     = {10.1103/PhysRevResearch.6.033034}
}

@article{Gaikwad2025GDQST,
  title   = {Gradient-Descent Methods for Fast Quantum State Tomography},
  author  = {Gaikwad, Akshay and Torres, Manuel Sebastian and Ahmed, Shahnawaz and Kockum, Anton Frisk},
  journal = {Quantum Sci. Technol.},
  volume  = {10},
  number  = {4},
  pages   = {045055},
  year    = {2025},
  doi     = {10.1088/2058-9565/ae0baa}
}

@misc{Cai2025OnlineSGDQST,
  title         = {Online Quantum State Tomography via Stochastic Gradient Descent},
  author        = {Cai, Jian-Feng and Jiao, Yuling and Li, Yinan and Lu, Xiliang and Yang, Jerry Zhijian and You, Juntao},
  howpublished  = {arXiv:2507.07601},
  year          = {2025},
  eprint        = {2507.07601},
  archivePrefix = {arXiv},
  primaryClass  = {quant-ph},
  doi           = {10.48550/arXiv.2507.07601}
}

@article{Cramer2010EfficientQST,
  title   = {Quantum State Tomography for Matrix Product Density Operators},
  author  = {Qin, Zhen and Jameson, Casey and Gong, Zhexuan and Wakin, Michael B. and Zhu, Zhihui},
  journal = {IEEE Trans. Inf. Theory},
  volume  = {70},
  number  = {7},
  pages   = {5030--5056},
  year    = {2024},
  doi     = {10.1109/TIT.2024.3360951}
}

@article{Baumgratz2013ScalableDensity,
  title   = {Quantum State Tomography with Locally Purified Density Operators and Local Measurements},
  author  = {Guo, Yuchen and Yang, Shuo},
  journal = {Commun. Phys.},
  volume  = {7},
  pages   = {322},
  year    = {2024},
  doi     = {10.1038/s42005-024-01813-4}
}

@misc{Cai2026MPOQST,
  title         = {Online {Riemannian} Gradient Descent for Quantum State Tomography with Matrix Product Operators},
  author        = {Cai, Jian-Feng and Li, Jingyang and Zhang, Xiaoqun and Zhang, Yuanwei},
  howpublished  = {arXiv:2605.04533},
  year          = {2026},
  eprint        = {2605.04533},
  archivePrefix = {arXiv},
  primaryClass  = {quant-ph},
  doi           = {10.48550/arXiv.2605.04533}
}

@article{Huang2020ClassicalShadows,
  title   = {Predicting Many Properties of a Quantum System from Very Few Measurements},
  author  = {Huang, Hsin-Yuan and Kueng, Richard and Preskill, John},
  journal = {Nat. Phys.},
  volume  = {16},
  number  = {10},
  pages   = {1050--1057},
  year    = {2020},
  doi     = {10.1038/s41567-020-0932-7}
}

@misc{Torlai2018NNQST,
  title         = {Parametric Quantum State Tomography with {HyperRBMs}},
  author        = {Tonner, Simon and Tran, Viet T. and Kueng, Richard},
  howpublished  = {arXiv:2601.20950},
  year          = {2026},
  eprint        = {2601.20950},
  archivePrefix = {arXiv},
  primaryClass  = {quant-ph},
  doi           = {10.48550/arXiv.2601.20950}
}

@article{Alquier2013RankPenalizedQST,
  title   = {Rank-Penalized Estimation of a Quantum System},
  author  = {Alquier, Pierre and Butucea, Cristina and Hebiri, Mohamed and Meziani, Katia and Morimae, Tomoyuki},
  journal = {Phys. Rev. A},
  volume  = {88},
  number  = {3},
  pages   = {032113},
  year    = {2013},
  doi     = {10.1103/PhysRevA.88.032113}
}

@article{Guta2012RankModelQST,
  title   = {Quantum Information Criteria for Model Selection in Quantum State Estimation},
  author  = {Yano, Hiroshi and Yamamoto, Naoki},
  journal = {J. Phys. A: Math. Theor.},
  volume  = {56},
  number  = {40},
  pages   = {405301},
  year    = {2023},
  doi     = {10.1088/1751-8121/acf747}
}

@article{Burer2003LowRankSDP,
  title   = {A Nonlinear Programming Algorithm for Solving Semidefinite Programs via Low-Rank Factorization},
  author  = {Burer, Samuel and Monteiro, Renato D. C.},
  journal = {Math. Program., Ser. B},
  volume  = {95},
  number  = {2},
  pages   = {329--357},
  year    = {2003},
  doi     = {10.1007/s10107-002-0352-8}
}

@inproceedings{Hazan2008SparseSDP,
  title     = {Sparse Approximate Solutions to Semidefinite Programs},
  author    = {Hazan, Elad},
  booktitle = {LATIN 2008: Theoretical Informatics},
  series    = {Lecture Notes in Comput. Sci.},
  volume    = {4957},
  pages     = {306--316},
  year      = {2008},
  publisher = {Springer},
  doi       = {10.1007/978-3-540-78773-0_27}
}

@inproceedings{Jaggi2013FrankWolfe,
  title     = {Revisiting {Frank--Wolfe}: Projection-Free Sparse Convex Optimization},
  author    = {Jaggi, Martin},
  booktitle = {Proc. 30th Int. Conf. Mach. Learn.},
  series    = {Proc. Mach. Learn. Res.},
  volume    = {28},
  number    = {1},
  pages     = {427--435},
  year      = {2013},
  publisher = {PMLR}
}

@article{Lanczos1950Iteration,
  title   = {An Iteration Method for the Solution of the Eigenvalue Problem of Linear Differential and Integral Operators},
  author  = {Lanczos, Cornelius},
  journal = {J. Res. Natl. Bur. Stand.},
  volume  = {45},
  number  = {4},
  pages   = {255--282},
  year    = {1950},
  doi     = {10.6028/jres.045.026}
}

@book{Saad2011NumericalMethods,
  title     = {Numerical Methods for Large Eigenvalue Problems},
  author    = {Saad, Yousef},
  edition   = {Revised},
  publisher = {SIAM},
  address   = {Philadelphia, PA, USA},
  year      = {2011},
  doi       = {10.1137/1.9781611970739}
}

@book{NielsenChuang2010QuantumComputation,
  title     = {Quantum Computation and Quantum Information},
  author    = {Nielsen, Michael A. and Chuang, Isaac L.},
  edition   = {10th anniversary},
  publisher = {Cambridge Univ. Press},
  address   = {Cambridge, U.K.},
  year      = {2010},
  doi       = {10.1017/CBO9780511976667}
}

@article{Hradil1997QuantumStateEstimation,
  title   = {Quantum-State Estimation},
  author  = {Hradil, Zden{\v{e}}k},
  journal = {Phys. Rev. A},
  volume  = {55},
  number  = {3},
  pages   = {R1561--R1564},
  year    = {1997},
  doi     = {10.1103/PhysRevA.55.R1561}
}

@misc{taherpour2026distributed,
  title         = {Distributed Learning of Quantum State Tomography Robust to Readout Errors},
  author        = {Taherpour, Amirhossein and Sadeghi, Alireza and Giannakis, Georgios B.},
  howpublished  = {arXiv:2604.14428},
  year          = {2026},
  eprint        = {2604.14428},
  archivePrefix = {arXiv},
  primaryClass  = {quant-ph},
  doi           = {10.48550/arXiv.2604.14428}
}
}

\end{document}